\documentclass[conference]{IEEEtran}
\IEEEoverridecommandlockouts

\usepackage{booktabs} 
\usepackage{times} 
\usepackage{xcolor} 
\usepackage{balance} 
\usepackage{url} 
\usepackage{tabularx}
\usepackage{siunitx} 
\usepackage{epsfig,tabularx,subfigure,multirow, graphicx}
\usepackage{enumitem}
\usepackage{caption}
\usepackage{pifont} 
\usepackage{makecell}
\usepackage{amsthm,amsmath}  
\usepackage{amsfonts}
\newcolumntype{L}[1]{>{\raggedright\arraybackslash}p{#1}}
\newcolumntype{C}[1]{>{\centering\arraybackslash}p{#1}}
\newcolumntype{R}[1]{>{\raggedleft\arraybackslash}p{#1}}

\usepackage[linesnumbered,ruled,vlined]{algorithm2e}
\SetKwRepeat{Do}{do}{while}
\SetCommentSty{mycommfont}

\long\def\comment#1{}

\DeclareMathOperator*{\argmin}{arg\,min}

\newcommand{\nop}[1]{}

\newcommand{\figureBelowMargin}{\vspace{0ex}}

\newtheorem{theorem}{\bf Theorem}[section]

\theoremstyle{remark}

\theoremstyle{definition}
\newtheorem{definition}{\bf Definition}

\newcommand{\revision}[1]{\color{black}{#1} \color{black}}
\def\BibTeX{{\rm B\kern-.05em{\sc i\kern-.025em b}\kern-.08em
		T\kern-.1667em\lower.7ex\hbox{E}\kern-.125emX}}

\begin{document}
	\title{Wait to be Faster: a Smart Pooling Framework for Dynamic Ridesharing}

	\author{
		{Xiaoyao Zhong{\small$~^{*}$}, Jiabao Jin{\small$~^{*}$}, Peng Cheng{\small$~^{*}$}, Wangze Ni{\small$~^{\Diamond}$}, Libin Zheng{\small$~^{\dagger}$}, Lei Chen{\small$~^{\Diamond}$},  Xuemin Lin{\small$~^{\ddagger}$}}\\
		\fontsize{10}{10}\itshape
		$~^{*}$East China Normal University, Shanghai, China\\
		\fontsize{10}{10}\itshape
		$~^{\Diamond}$HKUST(GZ) and HKUST, Guangzhou and Hong Kong SAR, China\\
		\fontsize{10}{10}\itshape
		$~^{\dagger}$Sun Yat-sen University, Guangzhou, China\\
		\fontsize{10}{10}\itshape
		$~^{\ddagger}$Shanghai Jiaotong University, Shanghai, China\\
		
		\fontsize{9}{9}\upshape
		\{xiaoyao.zhong, jiabaojin\}@stu.ecnu.edu.cn; pcheng@sei.ecnu.edu.cn; wniab@cse.ust.hk; \\ zhenglb6@mail.sysu.edu.cn; leichen@cse.ust.hk; xuemin.lin@gmail.com
	}

\maketitle

\begin{abstract}
	Ridesharing services, such as Uber or Didi, have attracted considerable attention in recent years due to their positive impact on environmental protection and the economy. Existing studies require quick responses to orders, which lack the flexibility to accommodate longer wait times for better grouping opportunities. In this paper, we address a NP-hard ridesharing problem, called \textit{Minimal Extra Time RideSharing} (METRS), which balances waiting time and group quality (i.e., detour time) to improve riders' satisfaction. To tackle this problem, we propose a novel approach called WATTER (WAit To be fasTER), which leverages an order pooling management algorithm allowing orders to wait until they can be matched with suitable groups. The key challenge is to customize the extra time threshold for each order by reducing the original optimization objective into a convex function of threshold, thus offering a theoretical guarantee to be optimized efficiently. We model the dispatch process using a Markov Decision Process (MDP) with a carefully designed value function to learn the threshold. Through extensive experiments on three real datasets, we demonstrate the efficiency and effectiveness of our proposed approaches.

\end{abstract}

\section{Introduction}
\label{sec:introduction}
With the increasing popularity of the sharing economy, more and more ridesharing platforms are emerging to facilitate people's lives, such as Uber and Didi. The ridesharing service is to group riders with overlapping travel routes and similar time schedules, and then assign them to workers to serve. It not only lowers prices for riders and saves fuel consumption for workers, but also eases traffic congestion and reduces carbon dioxide emissions.

In ridesharing, the platform may have different optimization goals: (1) maximizing platform revenue~\cite{cheng2019queue, zeng2020gas, wang2022revenue}; (2) minimizing the total travel distance of workers~\cite{xu2019insertion, liu2020mobility, Haliem2021modelfree}; (3) maximizing the number of served orders~\cite{santos2013dynamic, zeng2020gas}. To achieve a high service rate, orders will be dispatched even if they result in worse satisfaction, and will only be rejected when they cannot be served in the extreme cases.

In dynamic ridesharing, existing studies propose two processing modes: online-based mode and batch-based mode. Online-based methods~\cite{cheng2017utility, tong2018prunegdp, wang2022meeting} provide a real-time response to each order.  Batch-based methods~\cite{huang2014kinetic, bei2018algorithms, zeng2020gas} usually group the orders within a batch (i.e., a time window of 5 seconds) based on specific combination strategy and assign the groups to workers. We observe from real datasets that orders can wait for a while  (e.g., 10 seconds) to get a better grouping result with less travel costs. We illustrate this with the following example:

\begin{figure}[t]\centering
	\scalebox{0.16}[0.16]{\includegraphics{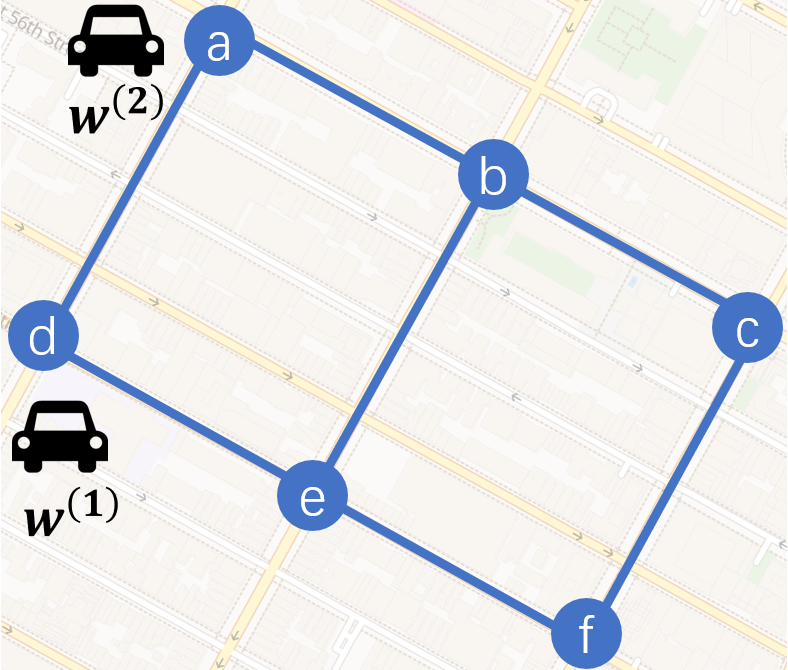}}
	\caption{\small An example of road network.}
	\label{fig:example1_roadnetwork}
\end{figure}

\begin{table}[t]
	\begin{center}
		{\revision{
				\small\scriptsize
				\caption{\small Online arriving orders.} \label{tab:example1}
				\begin{tabular}{p{1.1cm}|p{1.1cm}p{1.1cm}p{1.1cm}p{1.1cm}}
					\hline 
					{\bf Order} & {\bf Release Time (s)} & {\bf Batch round} & {\bf Pick-up Location} & {\bf Drop-off Location} \\
					$o^{(i)}$ & $t^{(i)}$ & $t^{(i)}/10$ & $l_p^{(i)}$ & $l_d^{(i)}$ \\ \hline 
					$o^{(1)}$ & $5$ & $0$ & $a$ & $c$ \\
					$o^{(2)}$ & $8$ & $0$ & $d$ & $f$ \\
					$o^{(3)}$ & $10$ & $0$ & $d$ & $c$ \\
					$o^{(4)}$ & $12$ & $1$ & $e$ & $f$ \\
					\hline
				\end{tabular}
			}
		}\vspace{-2ex}
	\end{center}
\end{table}

\par
\revision{\textit{
		Example 1. We assume that there are 2 idle workers $w^{(1)} \sim w^{(2)}$ and 4 orders $o^{(1)} \sim o^{(4)}$ arrive at the platform in ascending order of $t^{(i)}$. The optimization objective is to minimize the total travel time of the workers. The road network consists of 6 nodes and 7 edges. Each edge represents a road with a travel time of 1 minute. The information of orders is shown in the Table \ref{tab:example1}.
		(1) For the non-sharing method, the 2 available workers can only serve orders sequentially. The trajectory of $w^{(1)}$ is $\langle d, f, e, f\rangle$ and that of $w^{(2)}$ is $\langle a, c, d, c\rangle$. This method result in a total travel time of $4 + 8 = 12$ minutes.
		(2) For the online-based method, the platform will insert the locations into workers' routes greedily. The trajectory of $w^{(1)}$ is $\langle d, e, f, d, c\rangle$ and that of $w^{(2)}$ is $\langle a, c \rangle$. This method results in a total travel time of $7 + 2 = 9$ minutes.
		(3) For the batch-based method, let the batch size be $10$ seconds. For instance, orders $o^{(1)}\sim o^{(3)}$ are in batch round $0$, and order $o^{(4)}$ is in another batch round $1$. As a result, $o^{(1)}$ and $o^{(3)}$ will be grouped together. Because $o^{(2)}$ and $o^{(4)}$ are not in the same batch, they will be served sequentially. This method results in a total travel time of $4 + 3 = 7$ minutes.
		However, the best match for $o^{(1)}$ is actually $o^{(3)}$, and the best match for $o^{(2)}$ is $o^{(4)}$. Hence, resulting in a total travel time of $2+3=5$ minutes. Compared to the previous three methods, this pooling-then-grouping strategy only causes the orders to wait slightly longer but greatly reduces the total travel time.
}}

\textit{Challenge}: Unlike existing studies that only respond orders immediately or in a static mini-batch time, \textit{is it possible to allow orders to wait for a period of time to take advantage of better grouping opportunities, which would ultimately result in shorter total travel times?} Intuitively, the longer an order waits and the more other orders arrive, the higher probability of it being grouped with more suitable orders, which leads to that its total travel cost/time can be reduced. As it is difficult to directly predict the arriving orders in the next several seconds, the main challenge is to determine the optimal waiting/pooling time before dispatching the order.

To address the challenge, we formalize a new problem \underline{M}inimal \underline{E}xtra \underline{T}ime \underline{R}ide \underline{S}haring (METRS), which takes the waiting times and detour times into consideration. 

We propose a novel framework called WATTER (WAit To be fasTER), which leverages an \textit{order pooling management algorithm} to maintain the orders and the shareability relationships in the temporal shareability graph. We propose an effective average extra time threshold-based grouping strategy that assigns a threshold of \textit{expected extra time} to each order. We theoretically prove that the optimization objective of METRS problem can be reduced to a convex function of extra time threshold, providing a theoretical guarantee for optimization effectiveness.
We take spatio-temporal environment into consideration, then model the decision-making process of holding or dispatching orders in the pool as a Markov Decision Process (MDP). Historical data is used to offline generate training experience by simulating the dispatch process of the framework incorporated with the proposed grouping strategy. We utilize this experience to train the value function in MDP, which is then used as an estimation of the expected extra time in the online decision-making process.
To summarize, we make the following contributions in the paper:
\begin{itemize}[leftmargin=*]
	\item We formulate the METRS problem to balance waiting response time and detour time and prove its hardness in Section \ref{sec:problemDefinition}.  
	\item We introduce the order pooling management algorithm and related algorithms in Section \ref{sec:solution1}.
	\item We devise an average extra time threshold-based grouping strategy with theoretical analysis in Section \ref{sec:solution2}. 
	\item We model the dispatch process and establish an offline reinforcement learning model combined with the online threshold-based strategy to make decisions in Section \ref{sec:solution3}. 
	\item Extensive experiments on real datasets is conducted in Section \ref{sec:experimental}. 
\end{itemize}

\section{Problem Definition}
\label{sec:problemDefinition}

\subsection{Preliminaries}

\begin{table}
	\revision{
		\centering \vspace{-4ex}
		{\small\scriptsize
			\caption{\small Symbols and Descriptions.} \label{table0}
			\begin{tabular}{l|l}
				{\bf Symbol} & {\bf \qquad \qquad \qquad\qquad\qquad Description} \\ \hline \hline
				
				$o^{(i)}$   & an order sent by rider\\
				
				$O$   & the set of all orders\\
				
				$g$   & a group of orders \\
				
				$t^{(i)}$   & the release time of order $o^{(i)}$\\
				
				$t_r^{(i)}$   & the response time of order $o^{(i)}$\\
				
				$t_d^{(i)}$   & the detour time of order $o^{(i)}$\\
				
				$t_e^{(i)}$   & the extra time of order $o^{(i)}$\\
				
				$\tau^{(i)}$   & the drop-off deadline of order $o^{(i)}$\\
				
				$p^{(i)}$   & the reject penalty of order $o^{(i)}$\\
				
				$l_p^{(i)}$  & the pick-up location of order $o^{(i)}$\\
				
				$l_d^{(i)}$   & the drop-off location of order $o^{(i)}$ \\
				
				$L$   & a route of ordered sequence of locations \\
				
				$cost(l_i, l_j)$ & the the shortest travel cost of two locations $l_i$ and $l_j$ \\
				
				$T(L)$ & the total travel cost of the route $L$\\
				
				$w^{(j)}$   & a worker \\
				
				$k^{(j)}$   & the vehicle capacity of the worker $w^{(j)}$\\
				\hline
				\hline
			\end{tabular}
		}\vspace{-2ex}
	}
\end{table}

\begin{definition} (Order)
	An order is denoted by $o^{(i)}= \langle l_p^{(i)}, l_d^{(i)},c^{(i)},t^{(i)}, \tau^{(i)}, \eta^{(i)} \rangle$. The order $o^{(i)}$ is released at timestamp $t^{(i)}$ and contains $c^{(i)}$ riders. The order asks to deliver riders from the pick-up location $l_p^{(i)}$ to the drop-off location $l_d^{(i)}$ before the deadline $\tau^{(i)}$. The platform needs to give response to order within waiting time limit $\eta^{(i)}$.
\end{definition}

An \textit{order group} is a collection of orders that can be represented as $g=\{ o^{(1)}, o^{(2)}, ... ,o^{(|g|)}\}$, where $|g|$ denotes the number of orders in the group. Note that the waiting time limit $\eta^{(i)}$ is a customized parameter just to indicate the preferred limit waiting time of $o^{(i)}$, which is not a constraint. In our paper, if $o^{(i)}$ has waited more than $\eta^{(i)}$ time, it should be dispatched immediately when there is a suitable group, Otherwise, it will be rejected.

\begin{definition} (Worker)
	A worker can be denoted by $w^{(j)} = \langle l^{(j)},$ $ k^{(j)}, a^{(j)} \rangle$, where $l^{(j)}$ is the worker's current location, $k^{(j)}$ is the vehicle capacity, and $a^{(j)}$ is  worker's availability.
\end{definition}

The availability $a^{(j)}$ can either be \textit{idle} (e.g., waiting for an assignment) or \textit{busy} (e.g., delivering an order group). In this paper, we assume that a worker can only deliver one order group at a time. 

Let $G^{(j)}$ denote the set of order groups that worker $w^{(j)}$ served in a day. Then, the set $G=\cup_{w^{(j)}\in W}G^{(j)}$ is composed of all served order groups. Given the set  $O$ of all orders, let  $O^+=\cup_{g \in G}\cup_{o^{(i)}\in g} o^{(i)} \subseteq O$ be the successfully served orders in $G$. Let $O^-=O-O^+$ denote the set of all rejected orders.

\begin{definition} (Route) 
	The route is an ordered sequence of locations denoted by $L = \langle l_1, l_2, ... , l_{|L|} \rangle$, where each $l_i$ represents a location on the road network.
\end{definition}

An order group $g$ can generate a route $L$ by aligning a sequence that includes the pick-up and drop-off locations of all orders in $g$. Then, the assigned worker $w^{(j)}$ travels to location $l_1$ and follows the route to serve the orders in the group. We use $L^{(i)}$ to denote the sub-route starting from $l_1$, passing through $l_p^{(i)}$, and ending at $l_d^{(i)}$. The travel cost of the route can be calculated as $T(L) = \sum_{k = 1}^{|L| - 1}cost(l_k, l_{k+1})$, where $cost$ is the shortest travel time of two locations.

\begin{definition} (Response Time)
	For a given order $o^{(i)}$, let $t^{(i)}_n$ denote the time when the platform notifies the grouping result, then the response time of $o^{(i)}$ is $t_r^{(i)}= t^{(i)}_n - t^{(i)}$, which refers to  the  waiting time from the order being released to being notified with assignment.
\end{definition}

In practice, riders have limited patience for waiting. Long response times may cause riders to cancel their orders, resulting in potential revenue loss for the platform. However, minimizing response time alone may lead to missing the potential future properer riders and results in long detours, which also can sacrifice the satisfactions of riders and drive them to choose other transportation methods or platforms. Thus, it is important to smartly balance the response times and detours without clearly knowing future orders.

\begin{definition} (Detour Time)
	For a given order $o^{(i)}$ in order group $g$, let $L$ be the generated route for $g$, the detour time $t_d^{(i)}$ of $o^{(i)}$ is denoted by $t_d^{(i)}= T(L^{(i)}) - cost(l_p^{(i)}, l_d^{(i)})$, which refers to the additional time cost incurred by sharing the route with other riders compared to the minimal shortest time $cost(l_p^{(i)}, l_d^{(i)})$.
\end{definition}

Both response time and detour time can be considered as the extra cost of riders in taking the ridesharing service, which is the major factor to affect the riders' satisfaction. In this paper, we define extra time as a unified metric to reflect riders' satisfaction.

\begin{definition} (Extra Time)
	The extra time $t_e^{(i)}$ of order $o^{(i)}\in g$ is defined as:
	\begin{equation}
		t_e^{(i)} = \alpha t_d^{(i)} +\beta t_r^{(i)}
	\end{equation}
	where $\alpha$ and $\beta$ are coefficients used for trade-off between the detour time $t_d^{(i)}$ and the response time $t_r^{(i)}$. 
\end{definition}

We offer the flexibility to adjust the weight in definition. By setting $\alpha=1$ and $\beta=1$, $t_e^{(i)}$ represent the real extra time of the rider, compared to the shortest travel time of the order. 

\subsection{Problem Definition}
We defined the \underline{M}inimal \underline{E}xtra \underline{T}ime \underline{R}ide\underline{S}haring problem:

\begin{definition} (METRS Problem)
	Given an online order set $O$ and a worker set $W$, the METRS problem is to find a set of shareable order groups $G$ for each order $o^{(i)}\in O$ and assign each group of orders with a suitable worker, such that total extra time of orders in the platform $\Phi(W,O)$ is minimized:
	\begin{equation}
		\min_{G} \Phi (W,O)=  \sum_{o^{(i)} \in O^+} t_e^{(i)} + \sum_{o^{(j)} \in O^-} p^{(j)} \label{equ:problem}
	\end{equation}
	\noindent where $p^{(j)}$ indicates the penalty of $o^{(j)}$ if it is rejected. An order group $g$ is \textit{shareable} if and only if it can generate a \textit{feasible} route $L$ with a available worker $w^{(j)}$ satisfying the following constraints:
	\begin{enumerate}
		\item Sequential constraint: $\forall o^{(i)} \in g$, it has $l_p^{(i)} = l_x\in L$ and $l_d^{(i)}= l_y\in L$, then $x < y$ must be satisfied. 
		\item Deadline constraint: $\forall o^{(i)} \in g$, then $t^{(i)} + t_r^{(i)} + T(L^{(i)})$ $ < \tau^{(i)} $ must be satisfied. 
		\item Capacity constraint: at any time, the number of riders in the vehicle cannot exceed its capacity.
	\end{enumerate}
	\label{def:problem}
\end{definition}

The penalty in the objective function \ref{equ:problem} indicates dissatisfaction after the rider has waited for a long time but not been served. 
Note that riders in order $o^{(i)}$ can wait for a maximum response time $\max t_r^{(i)} = \tau^{(i)} - t^{(i)} - cost(l_p^{(i)}, l_d^{(i)})$, since if the response time is longer than $\max t_r^{(i)}$ its deadline constraint must be violated. In order to keep consistency with served orders, we set the penalty as the maximum response time $p^{(i)} = \max t_r^{(i)}$.

\subsection{Hardness}
We prove that the METRS problem is NP-hard by a reduction from the Shared-Route Planning Query (SRPQ) problem \cite{zeng2020gas}, which has been proved as  an existing NP-hard problem.

\begin{theorem} (hardness of the METRS problem)
	The METRS problem defined in Definition \ref{def:problem} is NP-hard.\label{theo:np}
\end{theorem}
\begin{proof}
	We prove the theorem by a reduction from the SRPQ problem defined in \cite{zeng2020gas}, which has been proved to be an NP-hard problem.  The goal of SRPQ problem is to find, for each worker $w \in W$, a route $S_w$, such that the total revenue of the platform $OBJ(W,O)$
	$=\sum_{w\in W}\sum_{o\in S_w}p_r $ is maximized, where $p_r $ is the fare/payment of each order. We can rewrite the objective function of SRPQ as below:
	
	{\scriptsize\begin{align}
			\max &\  OBJ(W,O) 
			\Rightarrow  \max \  \sum_{o^{(i)}\in O^+}p_r 
			\Rightarrow  \max \  \sum_{o^{(i)}\in O}p_r - \sum_{o^{(j)}\in O^-}p_r \notag
	\end{align}}
	Due to $\sum_{o^{(i)}\in O}p_r$ is a constant, we can reduce the objective function SRPQ problem into: $\min \sum_{o^{(j)}\in O^-}p_r$
	
	Then, by setting the coefficient $\alpha = \beta = 0$ in $t_e$, and setting $p^{(j)}$
	$=p_r$, we show that the reduced SRPQ problem is equivalent to the METRS problem. That is, for a given SRPQ problem, we can reduce it into an instance of METRS problem. The SRPQ problem can be solved in polynomial time if and only if the METRS problem can be solved in polynomial time. Since the SRPQ problem has been proved to be NP-hard, METRS problem is also NP-hard.
\end{proof}

\section{Overview of WATTER Framework}
\label{sec:solution1}

\begin{figure*}[t!]\centering\vspace{-2ex}
	\scalebox{0.6}[0.6]{\includegraphics{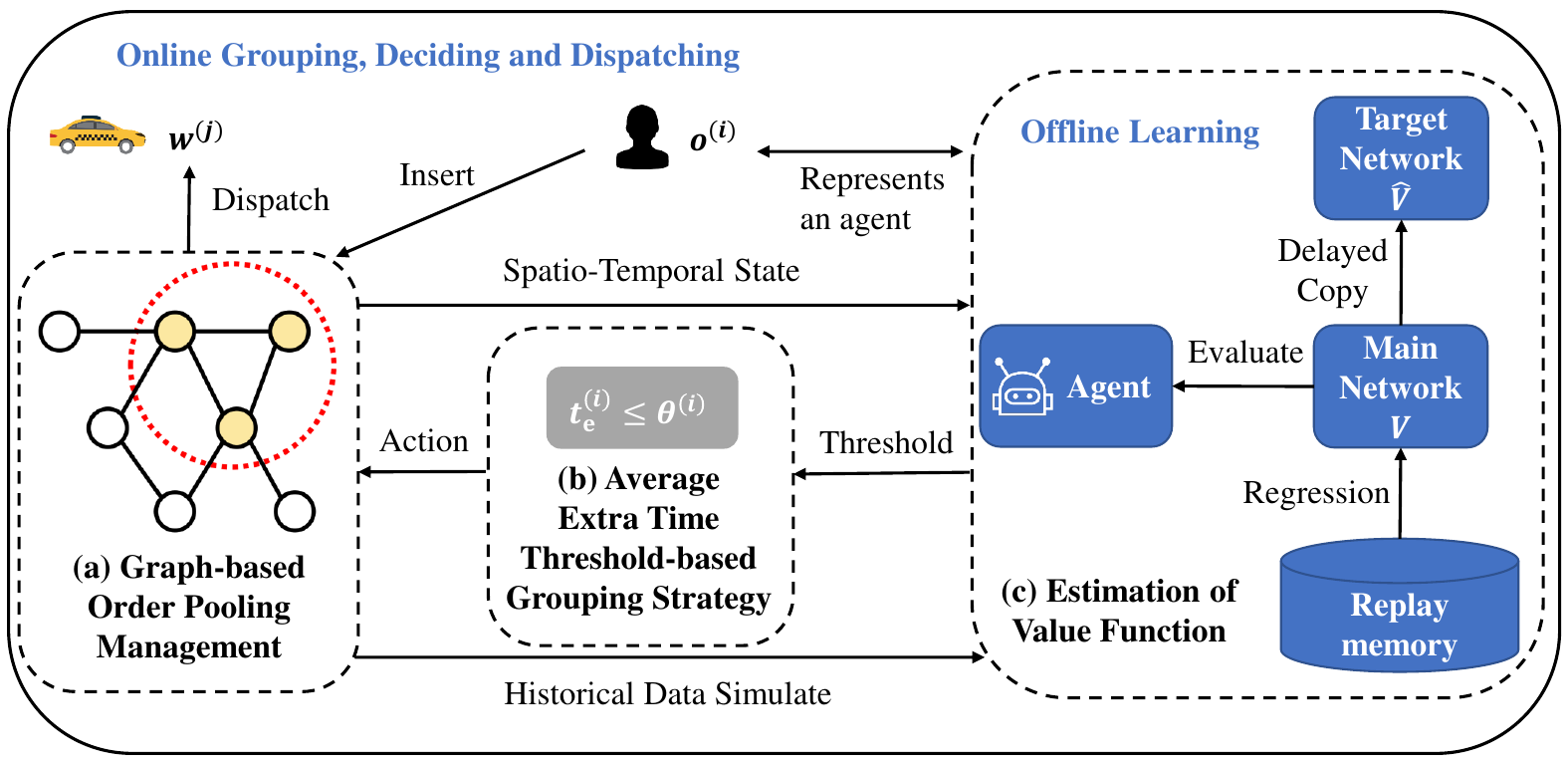}}
	\caption{\small Illustration of the WATTER framework.}
	\label{fig:overview}\vspace{-2ex}
\end{figure*}

We first introduce the three major parts of our WATTER framework: the order pooling management algorithm, the average extra time threshold-based grouping strategy and the estimation of value function for MDP reinforcement learning \cite{abbasi2019deeppool, shah2020adp, tang2021value} stage, as shown in Figure \ref{fig:overview}. During the online phase, we utilize a temporal shareability graph to dynamically manage temporal shareability relationships of orders. We conducted a theoretical analysis of the METRS problem to derive an average extra time threshold-based grouping strategy that utilizes the average extra time of each order and the customized threshold to make decisions. In the offline phase, we employ an MDP approach to estimate the value function, which is used as the threshold in the online decision-making process. 

\noindent \textit{(a) The Graph-based Order Pooling Management.}
Both online and batch approaches have the disadvantage of dispatching orders too quickly. Consequently, the matching pool is restricted to the currently available orders, ignoring possible future opportunities. To tackle these issues, we propose an order pooling management algorithm based on the temporal shareability graph that serves as a data structure to maintain a dynamic set of orders: each order is represented as a node, with edges connecting it to other orders that can be shared. This part involves maintaining the graph, identifying shareable order groups, and assign workers to order groups.

\noindent \textit{(b) The Average Extra Time Threshold-based Grouping Strategy.}
A simple but effective strategy is to dispatch an order when the extra time is below an expected threshold $\theta$. In reality, the expected threshold reflects the benefit that can be obtained from dispatching in the current spatiotemporal environment. We notice that a smaller threshold can lead to higher optimization results. However, a small threshold also prevents the order being grouped with more suitable orders.
In Section \ref{sec:solution2}, we analyze that through adjusting the expected threshold $\theta$, then we can obtain different optimization results. We transform the original METRS objective into a function of the expected threshold $\theta$, then we solve METRS by adjusting the expected threshold for each order. Fortunately, we prove that the reduced optimization objective possesses a convex shape, which is advantageous for solving. Then, we employ the gradient descent method to efficiently find the optimal threshold of $\theta$.

\noindent \textit{(c) The Estimation of Value Function.}
To reduce the optimization objective into a function on the expected threshold $\theta$, we previously utilized the distribution of extra times. However, the distribution may be impacted by spatiotemporal circumstances, such as peak periods, traffic congestion, supply of workers, and the demand for orders.  Given historical data, we can estimate the optimal expected threshold $\theta$~\cite{xu2018large, tang2021value}. 
We formulate the dispatch decision-making process as a Markov Decision Process (MDP), treating each order as an agent. Offline learning is achieved using components such as the main network, target network, and replay memory. However, the inherent learning process often faces challenges in converging effectively due to a lack of high quality learning experiences. To overcome this, we employ an off-policy training strategy that utilizes a threshold-based strategy to generate learning experiences and minimize the disparity between the result of MDP value function and the optimal expected threshold value. It allows us to fine-tune the threshold-based strategy based on different spatiotemporal environments, ultimately determining the desired expected threshold.

\section{Graph-Based Order Pooling Management}

\subsection{Temporal Shareability Graph}
\revision{
	In our WATTER framework, we enable orders to dynamically join or leave the order pool. However, existing batch-based methods ~\cite{bei2018bimatch, cheng2019queue, zeng2020gas} can only match orders within a batch, which cannot efficiently find cross-batch groups. To address this limitation, we introduce a temporal shareability graph as the order pool, which offers two key advantages. First, the match window can be dynamic and customized, unaffected by the batch size. Second, the pool maintains edges to represent the shareability relationships among orders, allowing efficiently filter out non-shareable order combinations, then to efficiently retrieve good sharing groups.
}

\begin{definition} (Temporal Shareability Graph)
	Given an order set $O$, the temporal shareability graph can be denoted by $\mathcal{G}  = (O, E)$. Each $o^{(i)} \in O$ represents a node in the graph, and the edge $e=(o^{(i)},o^{(j)}, \tau_e)$ $\in E$ represents that node $o^{(i)}$ and node $o^{(j)}$ can be shared in a group before timestamp $\tau_e$. 
\end{definition}

At any given timestamp, the snapshot of a graph provides the current shareability relationships among orders. Clique~\cite{zhang2019cliqueEnumeration} is a widely used cohesive subgraph structure for network analysis. A $k$-clique is a dense subgraph that has $k$ nodes, and where each pair of nodes are adjacent. Existing studies~\cite{chenyu2022pride, wuhan2022eride} have demonstrated that the shareability relationship is closely associated with the graph structure, as stated in Theorem \ref{theo:kclique}. Thus we can efficiently enumerate shareable groups by applying clique listing algorithms~\cite{yuanzhirong2022clique}.

\begin{theorem}
	Given a group $g$ containing $k$ orders, a feasible route $L$ can be generated only if the nodes corresponding to these $k$ orders in the shareability graph form a $k$-clique.
	\label{theo:kclique}
\end{theorem}

According to Theorem \ref{theo:kclique}, the shareable groups in the graph form cliques. The existing studies solely process the graph as a snapshot. However, our objective is to enable the graph to efficiently return the current k-cliques for decision-making. Thus, in this paper, we consider a \textit{temporal} shareability graph, which can supports insertions/deletions of nodes, and the expiration of edges. Given a group $g$ and a generated feasible route $L$, we use the $\tau_g$ to denote expiration time of the group, which is equal to the minimum slack time.

\begin{equation}
	\tau_g = \min_{o^{(i)}\in g} \tau^{(i)} - t^{(i)}  - T(L^{(i)}) - t_r^{(i)}
\end{equation}

\revision{
	The shareability graph is dynamically updated with the arrival and departure of orders. Nodes and corresponding edges in $\mathcal{G}$ are inserted or deleted accordingly. When a new order $o^{(i)}$ arrives, it is inserted as a new node. Then, we traverse $\mathcal{G}$ to find its neighbors that can be shared. \textit{The shareability relationship not only involves orders being close in terms of their release time, but also the proximity between pick-up and drop-off locations. Orders that are far apart in either time or space cannot be part of a feasible route.} When an order $o^{(j)}$ can be shared with $o^{(i)}$, we can find a group $g$ with feasible route $L$ with the minimal travel cost that can serve them together. A new edge $e=(o^{(i)}, o^{(j)}, \tau_g)$ is then inserted to denote the shareability relationship between the two orders.

	Note that the shareability graph only displays the shareability relationships among orders. It can be used to efficiently find shareable groups by enumerating $k$-cliques. Then we choose the group with the smallest average extra time among all possible shareable groups for each order as its best group. We use the best groups for decision-making and dispatching.
}

\subsection{The Order Pooling Management Algorithm}

To manage the online arriving orders, we introduce an order pooling management algorithm that processes orders in real-time. The orders are inserted into the order pool, which is maintained as a temporal shareability graph. Periodically, we check the current status of orders and dispatch them according to a specific strategy. If the order is not dispatched, it will stay in the pool, awaiting better grouping opportunities.

\revision{
	As shown in Algorithm \ref{algo:order_pool_algorithm}, we iteratively process new orders and insert them into the pool (lines 2-4). During the insertion process, we maintain the best group information of each order in $G_b$, which refers to the group that has the smallest average extra time among all the shareable groups that contain the order. We then remove any edges and groups that will expire after the current timestamp of system (lines 5-6). Next, we check all orders in the pool and retrieve the best group in the current spatiotemporal environment. Since we maintain a map of best group $G_b$ during pool updates, the time cost of retrieval operation is $O(1)$ (lines 8-9). Each order in the order pool has its own waiting time, but can also leave the pool early if certain conditions (e.g., Algorithm \ref{algo:expectation_based_make_decision} implements \textit{MakeDecision} ) are met (line 10). \textit{Note that here we use asynchronous periodic checks, which are performed periodically instead of after every insertion.} If we can find a valid group $g$, we assign it to the closest available worker (lines 11-13). If $o^{(i)}$ does not have a shareable group, it will remain in the pool and wait. If it exceeds the waiting time limit, we reject it (lines 14-16).
}

\begin{algorithm}[t]
	\DontPrintSemicolon
	\KwIn{A set $W$ of $m$ workers, a set $O$ of $n$ orders sorted by arriving time}
	\KwOut{The served groups $\mathbb{S}$ and the failed orders $\mathbb{F}$}
	\SetKwFunction{IsTimeout}{IsTimeout}\SetKwFunction{MakeDecision}{MakeDecision}
	initialize shareability graph $\mathcal{G}$ and best group map $G_b$\;
	\ForEach{new order $o^{(i)} \in O$} {
		insert $o^{(i)}$ into the pool\;
		update system current timestamp $t_s \leftarrow t^{(i)}$ \;
		
		\ForEach{$g, e$ that exipres after $t_s$} {
			remove orders in $g, e$ from the graph \;
		}
		
		$/*$ Asynchronous Periodicity Check Orders $*/$ \;
		\ForEach{$o^{(j)} \in \mathcal{G}$} {
			$g \leftarrow G_b[j] $  \;	
			\If{$g$ exists \textbf{and} \MakeDecision{$g$, $t^{(i)}$}} {
				assign the $g$ to a worker to serve. \;
				remove orders in group $g$ from the graph \;
				S.append($g$) \;
			} 
			\ElseIf{$o^{(j)}$ exceeds wait time limit} {
				remove order $o^{(j)}$ from the graph \;
				F.append($o^{(j)}$) \;
			}
		}
	}
	\Return{$\mathbb{S,F}$} \;
	\caption{Order Pooling Management Algorithm}
	\label{algo:order_pool_algorithm}
\end{algorithm}

In order to facilitate quick response of decision-maker, it is important to store the information of the current best group of each order in the pool. Each group corresponds to a $k$-clique in the shareability graph. When the graph is updated, the existing best groups may also need to be updated. There are four situations that can lead to updates of the graph: (1) order arrival (line 3); (2) order departure (line 12, 15); (3) edge expiration (line 6); and (4) group expiration (line 6). Due to the space limitation, please refer to Appendix A, B for the detail algorithms to handle them.

With the order pool, we enable customized waiting time for each order on the platform. A decision-maker is crucially needed to determine whether the current best group is sufficient for dispatch. For example, the \textit{online strategy} is to dispatch orders as early as possible, which notifies riders in the shortest time; the \textit{timeout strategy} is to dispatch orders as late as possible to obtain the best group opportunity.

\section{Average Extra Time Threshold-based Grouping Strategy}
\label{sec:solution2}

\subsection{Threshold-based Strategy}

As we discussed in Section \ref{sec:introduction}, existing solutions that respond to orders immediately or within a static mini-batch time can prevent riders from being grouped with potentially more suitable riders in the near future. Then, the platform may miss some good opportunities to reduce the extra time of riders. However, if the platform holds orders for too long, they may timeout. Thus, \textit{a smart decision strategy to hold or dispatch the orders is the key component of the platform, which is also the most challenging part of the METRS problem}.

To solve the decision problem, we can examine several case studies. For orders whose current best group is optimal, we can dispatch them. For orders that are difficult to group with others, we also need to dispatch them immediately, even if their current group quality is not optimal. Orders located in popular areas can continue to wait until better group results are available. Thus, whether to continue holding or dispatching immediately depends on the spatiotemporal environment of the orders, including their pick-up and drop-off locations, as well as their release time. 

We can use historical data to estimate the possible grouping results and calculate the expected threshold $\theta^{(i)}$ of extra time for each order $o^{(i)}$. We will introduce how to select a good $\theta^{(i)}$ for each order in Sections \ref{subsec:fittingOptimization} and \ref{sec:solution3}. The threshold $\theta^{(i)}$ can be considered as a reference: when the extra time $t_e^{(i)}$ of $o^{(i)}$ in a group arrangement is smaller than $\theta^{(i)}$, it means the group is better than the historical performance and $o^{(i)}$ should be dispatch with no more wait; otherwise, the order $o^{(i)}$ may wait for a better group in the future.

Based on the thresholds, we  propose a flexible decision strategy for holding or dispatching orders. As shown in Algorithm \ref{algo:expectation_based_make_decision}, we filter out the orders wait longer than the limit $\eta^{(i)}$ (lines 1-3). Those orders can be served when there are suitable workers, otherwise will be rejected. We use a strategy based on expected threshold to decide whether to dispatch. First, we calculate the average extra time of orders in group $g$ (line 4). Then, we use the average estimated threshold $\bar{\theta}$ as a reference to make the final decision (lines 5-6).

\begin{algorithm}[t]
	\DontPrintSemicolon
	\KwIn{An order group $g$, system current timestamp $t_s$}
	\KwOut{whether dispatching (True) or holding (False)}
	$t^{(i)} + \eta^{(i)}$ $\leftarrow$ the earliest timeout time of orders in $g$ \;
	\If{$t_s > t^{(i)} + \eta^{(i)}$} {
		\Return{True} \;
	}
	$\bar{t_e}$ $\leftarrow$ average extra time of order in $g$ \;	
	$\bar{\theta}$ $\leftarrow$ average expected threshold of orders in $g$\;	
	\Return{$\bar{t_e} \le \bar{\theta}$} \;
	\caption{Average Extra Time Threshold-based Grouping Strategy}
	\label{algo:expectation_based_make_decision}
\end{algorithm}

\subsection{Reduction of Problem}
In this section, we present a probability analysis of extra time. Our strategy is influenced not only by past orders, but also by the arrival probability of future ones. Therefore, we introduce the expectation of the original optimization objective $\mathbb{E}(\Phi(W,O))$ as our optimization goal, following the setting in many existing learning-based methods ~\cite{abbasi2019deeppool, tang2021value, ke2022RL4Delay}. 

For the objective function \ref{equ:problem} of the METRS problem $\Phi(W,O)$, its first part $\sum_{o^{(i)} \in O^+} t_e^{(i)}$ is the extra time for all served orders; its second part $\sum_{o^{(j)} \in O^-} p^{(j)}$ is the penalty received due to rejected orders. As we discussed in Definition \ref{def:problem}, $t_e^{(i)} \le p^{(i)}$ holds for all orders $o^{(i)}$. 
To minimize $\Phi(W,O)$, we need to find a suitable time to dispatch orders so that $t_e^{(i)}$ is as small as possible and the service rate $\mu = \frac{|O^+|}{|O|}$ of the platform is as high as possible. 

For each served order $o^{(i)}$, we introduce an indicator function $\mathbb{I}(i)$: $\mathbb{I}(i)=1$ means the group of  $o^{(i)}$ has average extra time smaller than its average expected threshold, that is $\bar{t_e} \le \bar{\theta}$; $\mathbb{I}(i)=0$ means $o^{(i)}$ has waited more than its limit wait time $\eta^{(i)}$. Thus, $O^+ = O^1 \cup O^0$, where $O^1$ is the set of the served orders with $\mathbb{I}(i)=1$ and $O^0$ is for the served orders with $\mathbb{I}(i)=0$. Then, the optimization objective can be rewritten as follows:\vspace{-3ex}

{\small\begin{align}
		&\min \mathbb{E}( \sum_{o^{(i)} \in O^+} t_e^{(i)} + \sum_{o^{(j)} \in O^-} p^{(j)} ) \notag\\
		\Rightarrow &\min \mathbb{E} ( \sum_{o^{(i)} \in O^1} t_e^{(i)}+ \sum_{o^{(k)} \in O^0} (t_d^{(k)} + \eta^{(k)} ) + \sum_{o^{(j)} \in O^-} p^{(j)} )  \notag\\
		\le  &\min \mathbb{E}( \sum_{o^{(i)} \in O^1} t_e^{(i)}+ \sum_{o^{(k)} \in O^0} p^{(k)} + \sum_{o^{(j)} \in O^-} p^{(j)} ) \notag\\ 
		\Rightarrow &\min  \mathbb{E} (\sum_{o^{(i)} \in O}\mathbb{I}(i)t_e^{(i)}+(1-\mathbb{I}(i))p^{(i)})
\end{align}}

When the order $o^{(i)}$ is dispatched by threshold (i.e., $\mathbb{I}(i)=1$), the platform incurs a loss of $t_e^{(i)}$, otherwise the platform at most incurs a loss of $p^{(i)}$ if $o^{(i)}$ waits more than limit or is rejected. Therefore, the optimization objective is to minimize the total loss incurred by all orders.

Intuitively, if we minimize the extra time $t_e^{(i)}$ for each served order, we can minimize $\Phi (W,O)$. By utilizing Algorithm \ref{algo:expectation_based_make_decision}, we have:
\begin{equation}
	\sum_{o^{(i)} \in g} t_e^{(i)} \le \sum_{o^{(i)} \in g} \theta^{(i)}= \bar{\theta} |g|
	\label{eq:strategy}
\end{equation}

\noindent where $\theta^{(i)} \in \Theta$ represents the expected threshold that we have selected for each order. We use $\bar{\theta}$ to denote average expected threshold for all orders in the current best shareable group. The extra time  $t_e^{(i)}$ may change as new orders join or old orders leave the pool. It cannot be directly constrained. However, because all orders are dispatched according to the dispatch strategy, we can establish an upper bound on the original optimization problem:\vspace{-2ex}

{\small\begin{align}
		& \min \mathbb{E} (\sum_{o^{(i)} \in O}\mathbb{I}(i)t_e^{(i)}+(1-\mathbb{I}(i))p^{(i)} ) \notag\\
		\le &  \min_{\Theta} \mathbb{E}(\sum_{o^{(i)} \in O}\mathbb{I}(i)\theta^{(i)}+(1-\mathbb{I}(i))p^{(i)})
\end{align}}

The reduced problem is a function of $\theta^{(i)}$ and $p^{(i)}$, where $p^{(i)}$ is determined by the order information and can be considered constant. Thus, the only variable we have is $\theta^{(i)}$. Under our strategy, larger threshold $\theta^{(i)}$ increases the probability of satisfying the decision condition and dispatching the order. We use $p(\mathbb{I}(i)=1)$ to denote the probability of the indicator function being 1. Let $f(x)$ be the probability density function of the distribution that $t_e^{(i)}$ follows, and let $F(x)$ be its cumulative distribution function. Then, we have:

{\small\begin{equation}
		p(\mathbb{I}(i)=1)=\int_0^{\theta^{(i)}}f(x)dx=F(\theta^{(i)}),
\end{equation}}

\noindent where $F(\theta^{(i)})$ measures the probability of each order $o^{(i)}$ being dispatched when its group's average extra time is smaller than its average expected threshold.

		Then, the problem can be rewritten as:
		{\small\begin{align}
				&  \min_{\Theta} \mathbb{E}(\sum_{o^{(i)} \in O}\mathbb{I}(i)\theta^{(i)}+(1-\mathbb{I}(i))p^{(i)})\notag \\
				\Rightarrow & \min_{\Theta} \sum_{o^{(i)} \in O} p(\mathbb{I}(i)=1)  \theta^{(i)}+(1-p(\mathbb{I}(i)=1))p^{(i)} \notag\\ 
				\Rightarrow & \min_{\Theta} \sum_{o^{(i)} \in O} F(\theta^{(i)})  \theta^{(i)} + ( 1 - F(\theta^{(i)}) ) p^{(i)} \notag\\
				\Rightarrow & \min_{\Theta} \sum_{o^{(i)} \in O} p^{(i)} - (p^{(i)} -  \theta^{(i)})F(\theta^{(i)}) \notag\\
				\Rightarrow & \max_{\Theta} \sum_{o^{(i)} \in O} (p^{(i)} - \theta^{(i)})F(\theta^{(i)})
				\label{eq:final_objective}
		\end{align}}

		Here, $p^{(i)}-\theta^{(i)}$ represents the minimum gain in the loss space that the platform can obtain after dispatching order $o^{(i)}$, which is a monotonically decreasing function of $\theta^{(i)}$. By setting $p^{(i)}=\tau^{(i)}- t^{(i)} - cost(l_p^{(i)},l_d^{(i)})$, then $p^{(i)}-\theta^{(i)}$ denotes the slack time of order in the group. A longer slack time means the order is dispatched to a more appropriate group in a shorter time. Since $F(\theta^{(i)})$ is a cumulative distribution function, it is monotonically increasing. Therefore, the product of these two functions must have a maximum value. In other words, the reduced objective function is a convex function of thresholds of orders. Our goal is to find the threshold $\theta^{(i)}$ that corresponds to the maximum value of this product function.

		\subsection{Distribution Fitting and Optimization} \label{subsec:fittingOptimization}
		As mentioned before, we treat the extra time $t_e$ as a random variable that follows a specific distribution. We then transform the optimization objective into a function of the expected threshold $\theta$. There are two key challenges: (1) determining the specific distribution of $t_e$, and (2) optimizing the objective function under the distribution assumption.
		
		\noindent \textbf{Distribution Fitting.} 
		Regarding the extra time of orders, we have several observations about its distribution: (1) Shorter orders may exhibit smaller extra time due to the difficulty of inserting detours along their routes; (2) Orders with both the pick-up and drop-off locations in popular areas may experience smaller extra time because there are more opportunities for sharing with similar orders; (3) Orders released during peak hours may result in smaller extra time as a large number of proper orders can be the group candidates. Thus, we can reasonably assume that the extra time of orders can be clustered, with many orders falling into the same sub-intervals. However, obtaining prior knowledge about the specific distributions for these clusters and the confidence level with which they adhere to these distributions is challenging.
		
		For random variables influenced by multiple factors, Gaussian Mixture Models (GMM) ~\cite{1999GMM} are commonly used. GMM incorporates multiple Gaussian distributions and combines them based on weights, making it suitable for fitting complex distributions. In the case of the extra time $t_e$, we can consider each influencing factor mentioned earlier as a sub-component of the GMM, collectively contributing to the overall distribution. The fitting of GMM can be accomplished using the Expectation-Maximization (EM) algorithm ~\cite{2001EM}, which is widely employed for this purpose.
		
		\noindent \textbf{Objective Optimization.} 
		Our reduced optimization objective is a convex function of the thresholds $\theta^{(i)}$ of orders $o^{(i)}$, ensuring the existence of a maximum value. 
		For convex functions and approximate convex functions, various methods such as Newton's method or Gradient Descent can be employed to find the optimal. While these methods may encounter local optima, for the optimization objective in this paper, only a few iterations are required to obtain the solution.
		
		As shown in Algorithm \ref{algo:distribution_fitting_and_optimization}, we firstly utilize historical data of the extra time as the distribution to be fitted. We then employ the EM algorithm to fit a Gaussian Mixture Model to this historical data (line 1). After obtaining the fitted distribution, we can easily calculate its cumulative distribution function $F$ (line 2). For each individual order $o^{(i)}$, we can express the function $g$ as $g(\theta) = (p^{(i)} - \theta)$ based on its penalty term (line 3-4). By combining these two functions, we can use existing optimization methods (e.g., Gradient Descent) to find the optimal expected threshold $\theta^{(i)}$ (line 5-6).

		\begin{algorithm}[t]
			\DontPrintSemicolon
			\KwIn{order set $O$, historical records of extra time $H$}
			\KwOut{optimal expected thresholds $\Theta$}
			
			$M$ $\leftarrow$ the GMM fitting result on $H$\;
			$F$ $\leftarrow$ the CDF of $M$ \;
			
			\ForEach{$o^{(i)} \in O$}{
				$g(\theta) = (p^{(i)} - \theta)$\;
				$\theta^{(i)}$ $\leftarrow$   $\argmin_{\theta} F(\theta)*g(\theta)$ \;
				$\Theta$ $\leftarrow$ $\Theta \cup \{\theta^{(i)} \}$ \;
			}
			
			\Return{$\Theta$} \;
			\caption{Distribution Fitting and Optimization}
			\label{algo:distribution_fitting_and_optimization}
		\end{algorithm}

		To summarize, we reduce the METRS problem to maximize a function of $\theta^{(i)}$. The core challenge of the reduced problem is estimating the expected threshold $\theta^{(i)}$ for each order $o^{(i)}$. Then, we utilize a Gaussian Mixture Model for distribution fitting and employ Gradient Descent to optimize the reduced objective. However, in practical situations, the optimal group of orders is usually dispatched quickly, then leads to no enough samples to accurately fit its distribution. While we can make assumptions about the distribution of $t_e$, the obtained solution is only a \textit{coarse-grained} result. The dispatch strategy is still heavily influenced by the spatiotemporal environment. Moreover, it is difficult to get the distribution information of new arrived orders in the future.

\section{Reinforcement Learning based Estimation}
\label{sec:solution3}

To overcome the shortcomings in Section \ref{subsec:fittingOptimization}, we can estimate the expected threshold based on a reinforcement learning approach \textit{offline} with historical data. Combined with the \textit{online} maintenance of the current best group, the dispatch strategy can be \textit{fine-grained} tailored for each order.

\subsection{Dispatch Strategy as MDP}

For the orders in the pool, they wait and go through multiple decisions until are either dispatched or expire. 
In this paper, we propose to model this process as a Markov Decision Process (MDP) from a local view  ~\cite{puterman1990markov,ke2022RL4Delay, abbasi2019deeppool}, where each individual order is modeled as an agent. The MDP captures the sequential decision process as an agent that observes the current \textit{environment}, takes an \textit{action} $a$, and \textit{transits} from state $s_t$ to $s_{t + \Delta t}$, then receives a certain \textit{reward} $r$. Here, we use $s_t$ to denote the time state, such as $s_0$ representing the initial state at time $0$. Next, we introduce the details of the \textit{state}, \textit{action}, \textit{state transition} and \textit{reward} of our MDP model.

\noindent \textit{State}. Follow the design in the existing study~\cite{ke2022RL4Delay}, we consider multiple spatio-temporal features in state $s_t$. These spatio-temporal features consist of two components: the basic feature and the environmental feature. The basic feature contains information such as the order's pick-up location and release timestamp. We use \textit{location index} and \textit{time index} to quantize the basic spatio-temporal feature of the order into and a number of regions and timeslots. The location index is obtained by dividing the examined city area into $n\times n$ region grids, which is commonly used in existing studies~\cite{ke2022RL4Delay, jiabao2022gridtuner}. The time index is obtained by dividing the time into intervals of $\Delta t$ seconds. The region information of the pick-up and drop-off locations of the order is represented as a vector $s_L$ using one-hot encoding. The timeslots of the order's release and wait are connected and fed into a two-dimensional vector $s_T$.

The environmental features include the current demand and supply distribution of the platform. We consider the current demand distribution of existing orders on the platform, including the distribution of pick-up locations and drop-off locations represented by vector $s_O$. The supply distribution of workers in each region is represented using vector $s_W$. The distribution vectors are all calculated by the location index. Combining this data, the spatiotemporal environmental state can be written as $s_t = [s_L, s_T, s_O, s_W]$.

\noindent  \textit{Action}. For each decision phase, there are two types of actions that an agent can perform. The \textit{dispatch} action with $a=1$ involves grouping the order and finding a worker to serve it. Specifically, the agent considers that its corresponding order has matched the desired group and can leave the pool. Another action is  \textit{wait} with $a=0$. Specifically, the \textit{wait} action means the agent thinks its corresponding order will be matched to a better group in the future.

\noindent  \textit{State Transition}. The sequential decision process involves waiting actions during the life cycle of an order. A wait action triggers a transition to the next state with the same position but a different time slot and environment. In the case of dispatch or expiration, the order will terminate its life cycle and receive a final reward. The dispatch action assigns the order to a worker along with the current best group, while expiration occurs in situations where the order has been waiting in the pool for too long and cannot be finished before its destination deadline. Expiration also can occur implicitly in decision process and is unobservable to the agent. Since the rider becomes more impatient, the order may be canceled at any time, which is also considered as an expiration for simplification.

\noindent  \textit{Reward}. The definition of reward $r_t$ depends on the optimization objective of the platform, which in this paper is transformed into Equation \ref{eq:final_objective}. Let $T$ be the time order to be dispatched or rejected. The key target to solve MDP is to learn the \textit{state value function} $V_\pi(s_t)=E(\sum_{i=1}^{T-t}\gamma^{i-1}r_{t+i})$ under the current strategy $\pi$ ~\cite{tang2021value, xu2018large, ke2022RL4Delay}, where $\gamma$ is the discount factor that controls how far the agent looks into the future for rewards. By using MDP to model the dispatch decision, the expected extra time threshold is corresponds to the ``state-value function''. The strategy $\pi$ used here is Algorithm \ref{algo:expectation_based_make_decision}. That is, $\theta^{(i)}=p^{(i)} - V_\pi(s^{(i)}_t)$, where $s^{(i)}_t$ is the spatio-temporal environment of the order $o^{(i)}$. The Bellman Update ~\cite{tang2021value, xu2018large} corresponding to each of the two different actions are
\begin{small}
	\begin{equation}
		V_\pi(s^{(i)}_t) \leftarrow \left\{
		\begin{aligned}
			&p^{(i)} - t_d^{(i)}  & a =1   \\
			&-\Delta t + \gamma^{\Delta t} V_\pi(s^{(i)}_{t+\Delta t})(1-\mathbb{I}(expired)) & a=0 
		\end{aligned}
		\right.\notag
	\end{equation}
\end{small}
The \textit{wait} action (i.e., $a = 0$) may result in two types of rewards: (1) continuing to wait in the next round, which has an immediate reward of $-\Delta t$ with indicator $\mathbb{I}(expired) = 0$. Only in this case, the state-value function is related to future states and rewards. (2) being expired, which has a reward of $0$ with indicator $\mathbb{I}(expired) = 1$. When the agent takes the \textit{dispatch} action (i.e., $a = 1$), the final reward will be the sum of a positive penalty reward $p^{(i)}$ and a negative detour time reward $t_d^{(i)}$ of the order in the current best group.

Although we set an immediate reward for each action, the actual reward accumulates over each decision phase. It can be calculated as $\sum_{t = 0}^{t_r^{(i)} / \Delta t} - \gamma^{t}\Delta t$. By setting the discount factor $\gamma=1$, the reward of expired order accumulated over the decision phases is equal to \revision{the negative waiting response time $- t_r^{(i)}$. Regarding dispatched orders,} the reward is calculated as the slack time $p^{(i)} - t_e^{(i)}$. A higher slack time indicates a higher quality group for the order. Therefore, the agent is trained to prioritize better group quality (i.e., more slack time) and to avoid order expiration (i.e., negative reward).

The accumulated reward $R$ can be calculated as follows:
\begin{small}
	\begin{equation}
		R(s_0^{(i)})= \left\{
		\begin{aligned}
			& -t_r^{(i)} + p^{(i)} - t_d^{(i)} = p^{(i)} - t_e^{(i)}  & dispatched   \\
			& -t_r^{(i)} = -\max t_r^{(i)} & expired 
		\end{aligned}
		\right.
	\end{equation}
\end{small}
where $s_{0}^{(i)}$ is the initial state of agent representing order $o^{(i)}$. \revision{Note that when using the value function in Algorithm \ref{algo:expectation_based_make_decision}, we calculate $\theta^{(i)}$ as $p^{(i)} - V_\pi(s_t^{(i)})$. It's because we estimate $V_\pi(s_t^{(i)})$ as $p^{(i)} - \theta^{(i)}$. }

\subsection{Deep-$Q$-Network Learning Approach}

The aim of MDP is to maximize the expectation of the accumulated reward for each agent. If an agent takes too many wait actions, its correspond order may expire and leads to a negative reward.  When dealing with decisions, the agent will carefully consider whether a wait action is needed in each decision phase to obtain a higher cumulative reward. Therefore, the total benefit of MDP is consistent with our optimization objective Equation \ref{eq:final_objective}.

We use the the neural network $V(s_t^{(i)})$ to represent the state-value function $V_\pi(s_t^{(i)})$, which is the estimation of the expected extra time threshold in Equation \ref{eq:strategy}. Then, we decide whether to dispatch and ask agents to get the maximum benefit during the training phase. As a result, the learned state-value function is the optimal expected threshold $\theta$. To achieve this, following existing value-based methods such as \textit{Q}-Learning~\cite{watkins1992q, hasselt2010double} and Deep \textit{Q}-Networks (DQN)~\cite{hasselt2016doubleQL, ke2022RL4Delay}, we use a \textit{replay memory} $M$ to store the experience tuples. There are two networks: (1) the \textit{main network} $V$, used to estimate the value function;  (2) the \textit{target network} $\hat{V}$, which is a delayed copy of $V$ and is used to stabilize the training process. Mean-squared Temporal-Difference (TD) Error is a commonly used loss function to estimate the value function: 
\begin{equation}
	loss_{td} (s) = (r_t + \gamma\hat{V}(s^{(i)}_{t + \Delta t}) - V(s_t^{(i)}))^2. \notag
\end{equation}

However, the TD loss alone is insufficient to meet our expectations for estimating the expected threshold, because the TD loss only guarantees the value relationship between states. We also require the value function to have a relationship with the  extra time $t_e$ so that it can be directly utilized in the threshold-based strategy. Therefore, we introduce the target loss to learn the threshold.
\begin{equation}
	loss_{tg} (s) = (p^{(i)} - \theta^{(i)} - V(s_t^{(i)}))^2, \notag
\end{equation}
where $\theta^{(i)}$ is the optimal threshold obtained through distribution fitting and optimization in Section \ref{sec:solution2}.

The TD loss guides the agent in selecting the action, either waiting or dispatching, that has a higher overall reward. This ultimately helps the value function estimate better thresholds. Meanwhile, the target loss aims to align the DQN's learning outcomes with the existing strategy, allowing for fine-tuning through actions based on the established policy. Consequently, our final loss is a weighted sum of these two parts.
\begin{equation}
	loss_{(s_t, a, s_{t + \Delta t}, r_t) \in M} (s) = \omega loss_{td} + (1 - \omega) loss_{tg}, \notag
\end{equation}
where $\omega$ is a weight parameter to balance $loss_{td}$ and $loss_{tg}$.

\section{Experimental Study}
\label{sec:experimental}

In this section, we present the experimental setup and results of our algorithms.

\subsection{Experimental Setup}

\begin{table}[t]\vspace{2ex}
	\begin{center}
		{\small\scriptsize
			\caption{\small Experimental Settings.} \label{tab:settings}
			\begin{tabular}{l|l}
				{\bf \qquad \qquad \quad Parameters} & {\bf \qquad \qquad \qquad Values} \\ \hline \hline
				\multirow{2}{*}{the number $m$ of riders} 
				& \quad (NYC) \quad \  50K, 75K,\textit{100K}, 125K\\ 
				& (CDC, XIA) 30K, \textit{40K}, 50K, 60K \\
				the number $n$ of workers  & 3K, 4K, \textit{5K}, 6K\\
				the deadline scale parameter $\tau$  & 1.2, 1.4, \textit{1.6}, 1.8 \\
				the maximum capacity of vehicles $K_w$ & 2, \textit{3}, 4, 5\\
				the balance parameter  $\alpha $ and $ \beta$  & \textit{1}\\
				\hline
			\end{tabular}
		}
	\end{center}
\end{table}

\noindent \textit{Data Set.}
We evaluate our algorithm on three real-world city datasets. The first is a public dataset collected from yellow taxis~\cite{nyc_dataset} in New York City (NYC), USA. The second and third are order datasets collected from the GAIA platform of Didi Chuxing~\cite{didi} in Chengdu (CDC) and Xi'an (XIA), China. We use order data from 1 to 06 and 08 to 31 July, 2013 for NYC, from 1 to 29 November, 2016 for CDC and from 1 to 30 October 2016 for XIA to train the value function proposed in Section \ref{sec:solution2}. In the experiments for evaluating parameters, we use order data from 07 July, 2013 in New York, 30 November, 2016 in Chengdu and 31 October 2016 in Xi'an. Each order in the dataset contains the longitude and latitude of the pick-up and drop-off locations as well as the release time of the order. The specific experiment-related parameters are shown in Table \ref{tab:settings} (the italic values are the default  parameters).

\noindent \textit{Implementation.} 
\revision{
	We simulate ridesharing in our framework using the following settings. Follow the grid index construction methods in existing studies~\cite{cheng2019queue, ke2022RL4Delay, jiabao2022gridtuner}, we partition the city into several cells to serve as a grid index to speed up workers and riders search. We tested the performance impact of different grid size and choose $ 10 \times 10$ cells as the setting of grid index. We set the watching window for each order (e.g., $\eta^{(i)} = \eta * cost(l_p^{(i)}, l_d^{(i)})$) and choose $\eta = 0.8$ as the default value. We also choose time slot size $\Delta t = 10$ (seconds) as the default value. (For a detailed analysis, please refer to Appendix D, F, G. For the implementation of the reinforcement learning model, please refer to Appendix C, E.} 

We treat each record as an order with one passenger and thus a $k$-clique in our framework can represent a group of orders with k riders. In addition, we set the deadline for each order (e.g., $\tau^{(i)} = t^{(i)} + \tau * cost(l_p^{(i)}, l_d^{(i)})$), which is a frequently used setup in numerous previous studies~\cite{tong2018prunegdp, xu2019insertion, zeng2020gas}. To achieve fairness, we generate the maximum passenger capacity parameter $K_w$ and the pick-up position for workers for the three datasets based on specific distributions. We uniformly sample initial locations for workers using the distribution of orders' pick-up locations. The vehicle capacity $k^{(j)}$ of worker is uniformly sampled within the range $[2, K_w]$.

All experiments are implemented in C++ and compiled using -O3 optimization. The experiments are conducted on a single server equipped with a Xeon Silver 4214 CPU@2.20GHz and 128 GB RAM. All algorithms run on a single thread.

\noindent \textit{Compared Algorithms.}
We compare our algorithm \textit{\textit{WATTER-expect}} with the following algorithms:
\begin{itemize}[leftmargin=*]
	\item \textit{\textit{WATTER-online.}} (this paper) A variant baseline implemented in the Order Pooling Management Algorithm that uses online strategy for dispatching orders. Each order is dispatched as early as possible. 
	\item \textit{\textit{WATTER-timeout.}} (this paper) Another variant baseline implemented in the Order Pooling Management Algorithm that uses timeout strategy for dispatching orders. Each order is dispatched as late as possible. 
	\item \textit{\textit{GDP}}~\cite{tong2018prunegdp}. An online-based algorithm where orders can only use information from existing orders. It greedily tries to insert the pick-up and drop-off locations of orders into the worker's route. 
	\item \textit{\textit{GAS}}~\cite{zeng2020gas}. A batch-based algorithm where orders within a batch are processed together. It generates an additive tree for all orders that can be served by each worker and finds the group with the maximum utility in the tree to dispatch.
\end{itemize}

\noindent \textit{Measurements. }
All algorithms are evaluated in terms of \textit{Extra Time(s)}, \textit{Unified Cost}~\cite{tong2018prunegdp}, \textit{Service Rate(\%)} ($|O^+| / |O|$) and \textit{Running Time(s)}. The \textit{Unified Cost} $UC$ is calculated by the sum of worker cost and penalty for rejected orders. Following the existing study ~\cite{tong2018prunegdp}, we set the balance parameter as $1$ in $UC$, and the penalty of each order as $10 \times cost(l_p^{(i)}, l_d^{(i)})$. The \textit{Running Time(s)} is the average algorithm running time of each order. Except for \textit{Extra Time(s)}, other three metrics are widely used in the existing large-scale ridesharing studies~\cite{tong2018prunegdp, zeng2020gas, wang2020demand}. We early terminate the algorithms that not completed experiments within 24 hours.

\subsection{Experimental Results}

\noindent \textit{Impact of Varying Number of Riders. } Figure \ref{fig:var_rider} presents the results of varying the number of orders. In all datasets, our proposed algorithms WATTER-expect, WATTER-online, and WATTER-timeout are more effective (i.e., lower unified cost and extra time) than existing approaches with the increase of the number of orders. 
For instance, when $n=50k$, WATTER-expect achieved 12.2\%, 18.4\%, 35.7\% and 40.1\% lower extra time compared to WATTER-online, WATTER-timeout, GAS and GDP in the CDC dataset, respectively.

\begin{figure}[t!]
	\subfigure{
		\scalebox{0.22}[0.21]{\includegraphics{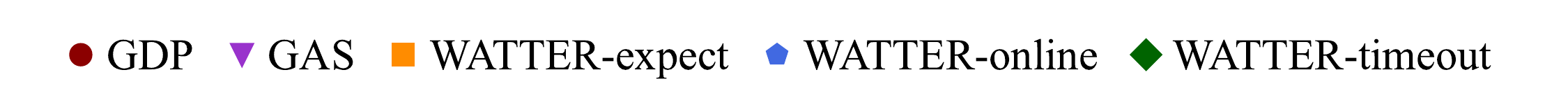}}}\hfill\\\vspace{-4ex}
	\addtocounter{subfigure}{-1}
	
	\subfigure[][{\scriptsize Extra Time(NYC)}]{
		\scalebox{0.22}[0.22]{\includegraphics{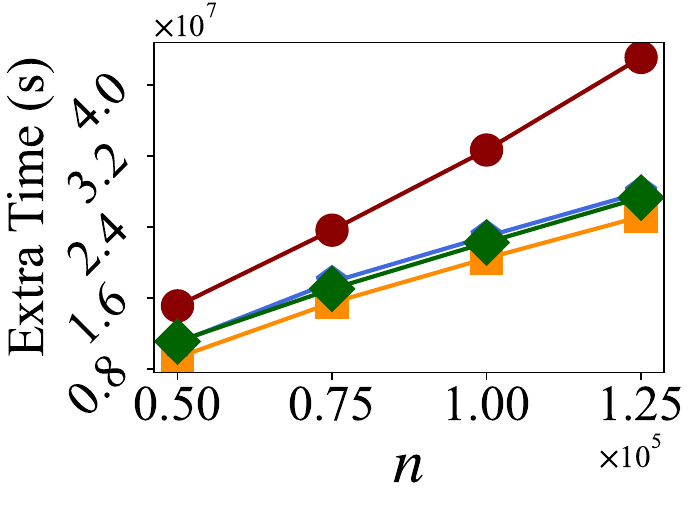}}
		\label{subfig:request_size_test_extra}}
	\hfill
	\subfigure[][{\scriptsize Extra Time(CDC)}]{
		\scalebox{0.22}[0.22]{\includegraphics{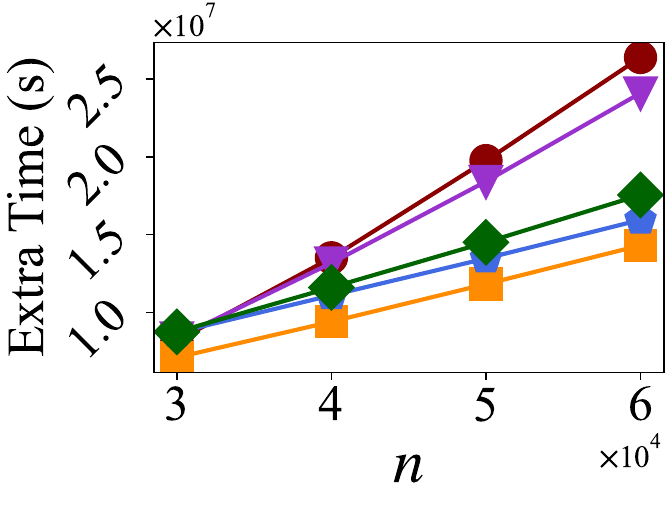}}
		\label{subfig:request_size_chengdu_extra}}
	\hfill
	\subfigure[][{\scriptsize Extra Time(XIA)}]{
		\scalebox{0.22}[0.22]{\includegraphics{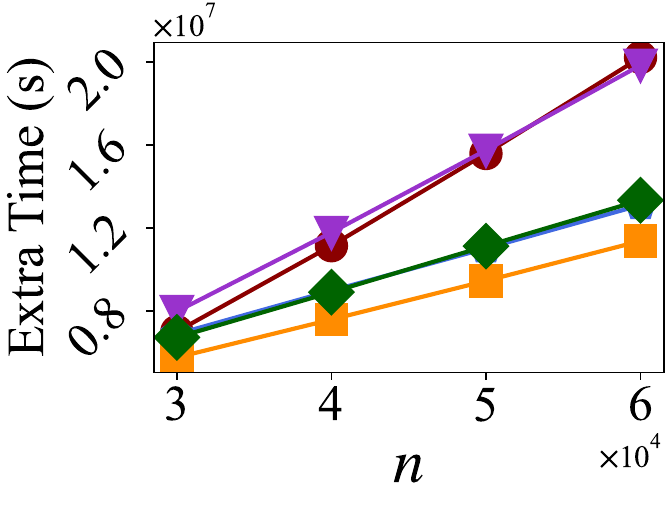}}
		\label{subfig:request_size_xian_extra}}\vspace{-2ex}
	
	\subfigure[][{\scriptsize Unified Cost(NYC)}]{
		\scalebox{0.22}[0.22]{\includegraphics{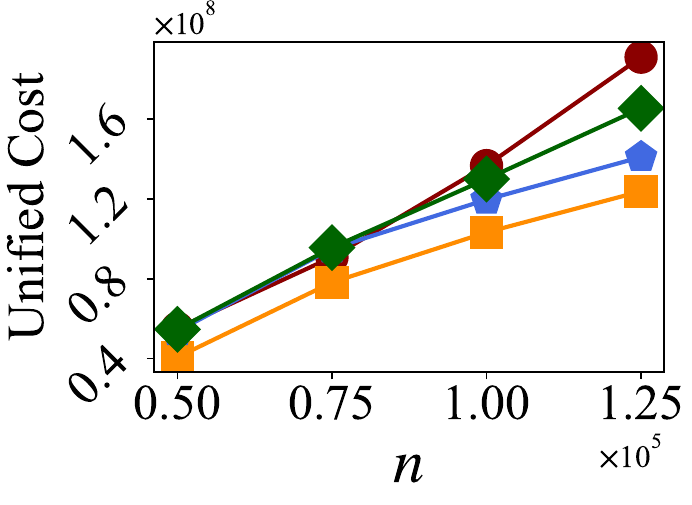}}
		\label{subfig:request_size_test_cost_t}}
	\hfill
	\subfigure[][{\scriptsize Unified Cost(CDC)}]{
		\scalebox{0.22}[0.22]{\includegraphics{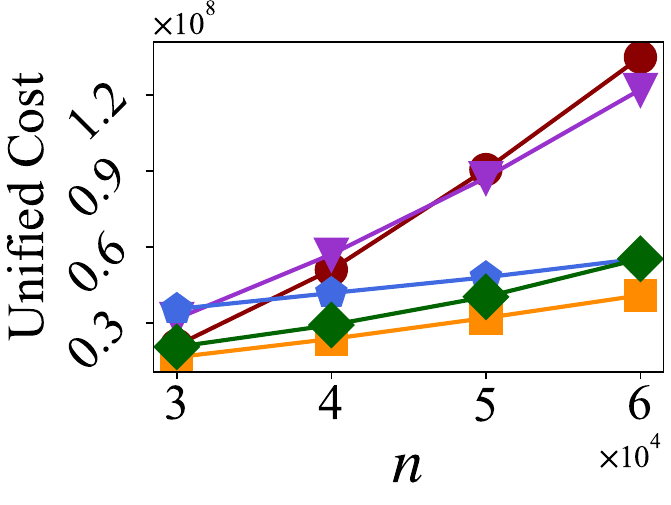}}
		\label{subfig:request_size_chengdu_cost_t}}
	\hfill
	\subfigure[][{\scriptsize Unified Cost(XIA)}]{
		\scalebox{0.22}[0.22]{\includegraphics{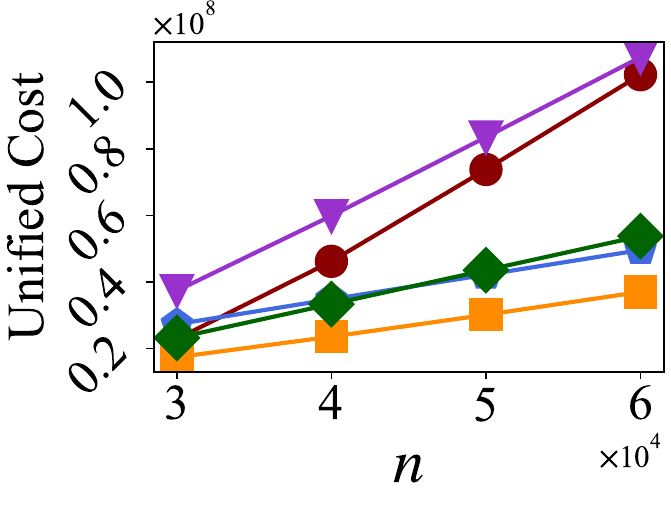}}
		\label{subfig:request_size_xian_cost_t}}\vspace{-2ex}
	
	\subfigure[][{\scriptsize Service Rate(NYC)}]{
		\scalebox{0.22}[0.22]{\includegraphics{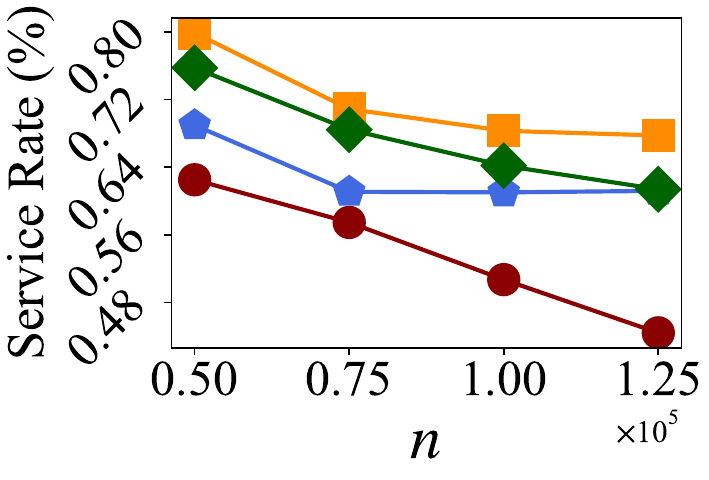}}
		\label{subfig:request_size_test_served_rate}}
	\hfill
	\subfigure[][{\scriptsize Service Rate(CDC)}]{
		\scalebox{0.22}[0.22]{\includegraphics{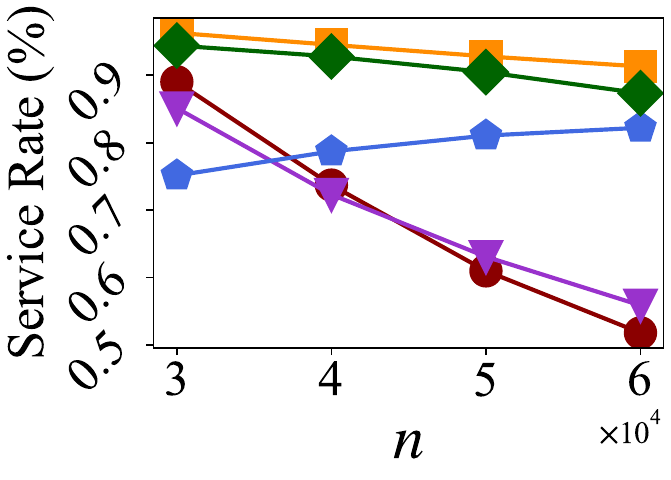}}
		\label{subfig:request_size_chengdu_served_rate}}
	\hfill
	\subfigure[][{\scriptsize Service Rate(XIA)}]{
		\scalebox{0.22}[0.22]{\includegraphics{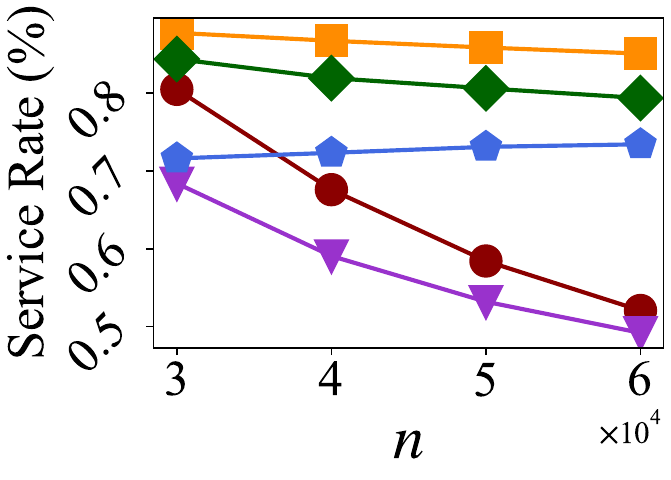}}
		\label{subfig:request_size_xian_served_rate}}\vspace{-2ex}
	
	\subfigure[][{\scriptsize Running Time(NYC)}]{
		\scalebox{0.22}[0.22]{\includegraphics{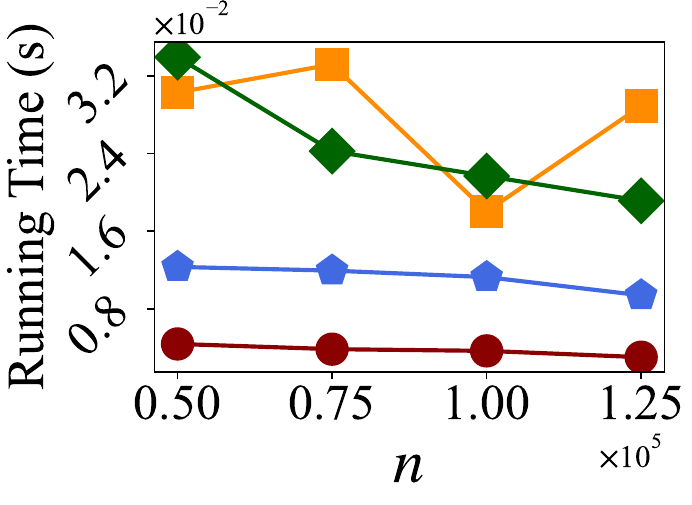}}
		\label{subfig:request_size_test_running_time}}
	\hfill
	\subfigure[][{\scriptsize Running Time(CDC)}]{
		\scalebox{0.22}[0.22]{\includegraphics{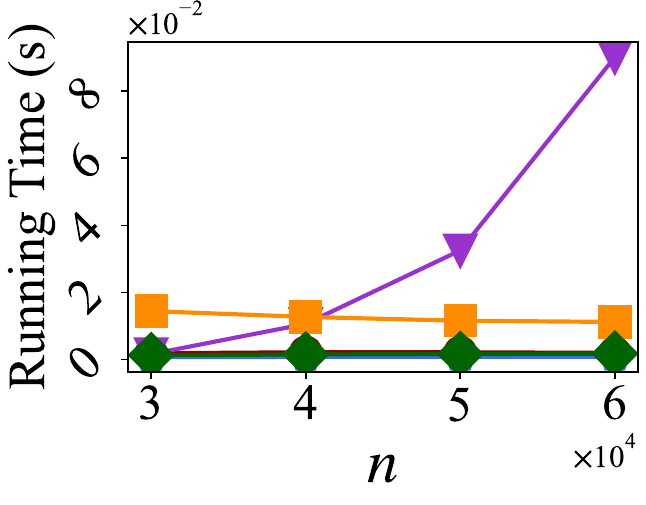}}
		\label{subfig:request_size_chengdu_running_time}}
	\hfill
	\subfigure[][{\scriptsize Running Time(XIA)}]{
		\scalebox{0.22}[0.22]{\includegraphics{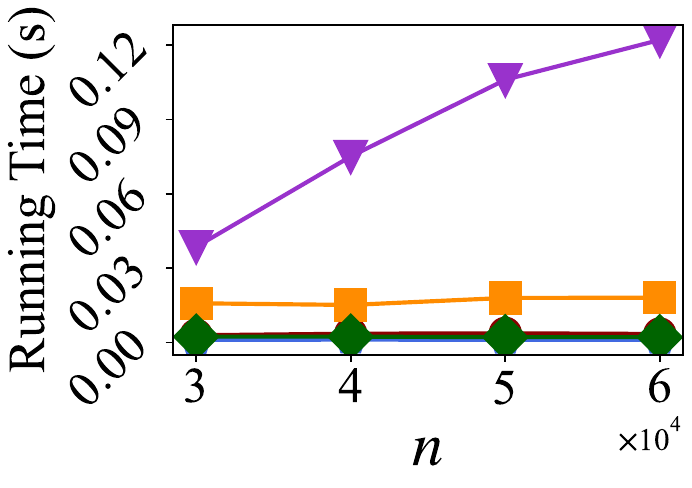}}
		\label{subfig:request_size_xian_running_time}}
	\caption{\small Performance of varying $n$.}\figureBelowMargin
	\label{fig:var_rider}
\end{figure}

Regarding service rate, our WATTER framework can outperform GDP and GAS. For instance, when $n=50k$, WATTER-expect achieved 2.3\%, 11.7\%, 36.9\% and 41.0\% improvement in service rate compared to WATTER-timeout, WATTER-online, GAS, and GDP in the CDC dataset, respectively. As for running time, GDP is the fastest algorithm due to its greedy insertion without enumerating possible order groups. GAS, however, has an exponentially increasing time cost due to an increase in the number of orders, thus is the slowest and cannot finish within 24 hours on NYC dataset. Among the WATTER algorithms, WATTER-online is faster than WATTER-timeout because it maintains a small enough order pool to minimize the insertion and update costs of new orders. Due to the need to use neural network, the running time of WATTER-expect is second-highest for most cases.

\begin{figure}[t!]
	\subfigure{
		\scalebox{0.22}[0.21]{\includegraphics{legend.eps}}}\hfill\\\vspace{-4ex}
	\addtocounter{subfigure}{-1}
	
	\subfigure[][{\scriptsize Extra Time(NYC)}]{
		\scalebox{0.22}[0.22]{\includegraphics{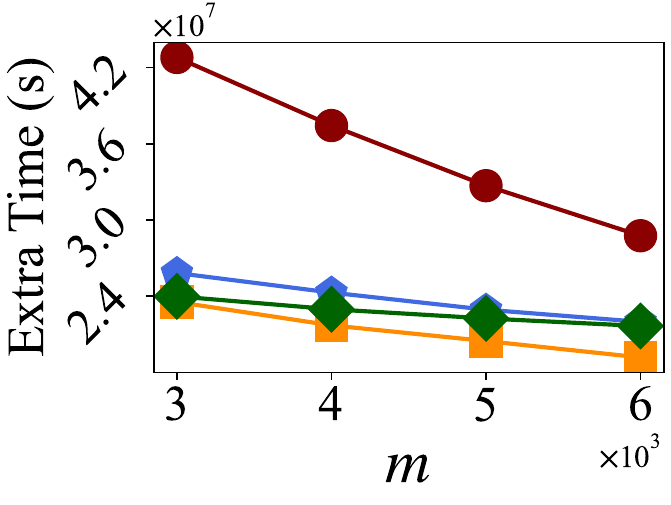}}
		\label{subfig:workers_test_extra}}
	\hfill
	\subfigure[][{\scriptsize Extra Time(CDC)}]{
		\scalebox{0.22}[0.22]{\includegraphics{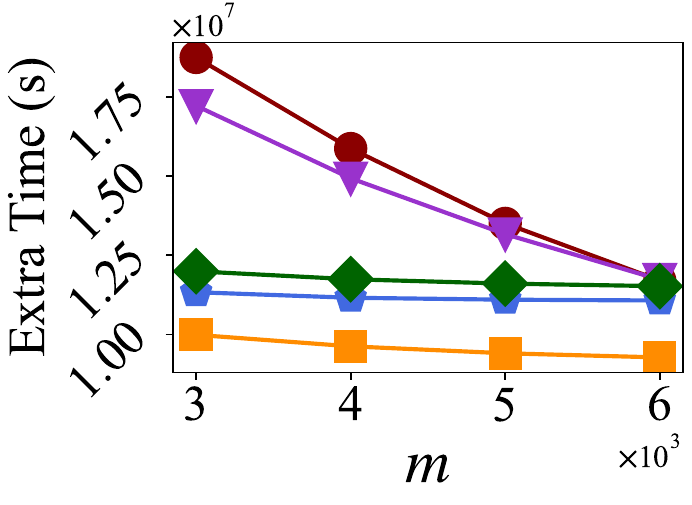}}
		\label{subfig:workers_chengdu_extra}}
	\hfill
	\subfigure[][{\scriptsize Extra Time(XIA)}]{
		\scalebox{0.22}[0.22]{\includegraphics{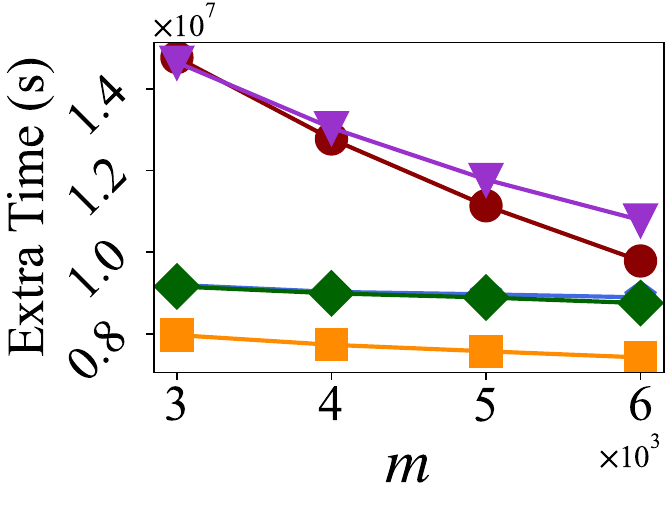}}
		\label{subfig:workers_xian_extra}}\vspace{-2ex}
	
	\subfigure[][{\scriptsize Unified Cost(NYC)}]{
		\scalebox{0.22}[0.22]{\includegraphics{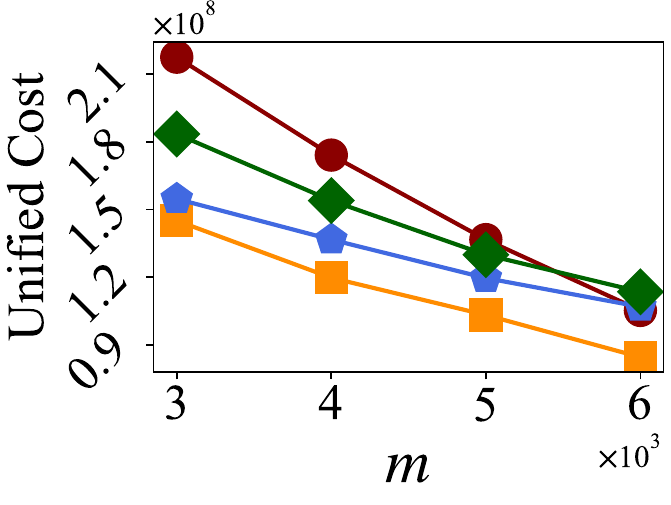}}
		\label{subfig:workers_test_cost_t}}
	\hfill
	\subfigure[][{\scriptsize Unified Cost(CDC)}]{
		\scalebox{0.22}[0.22]{\includegraphics{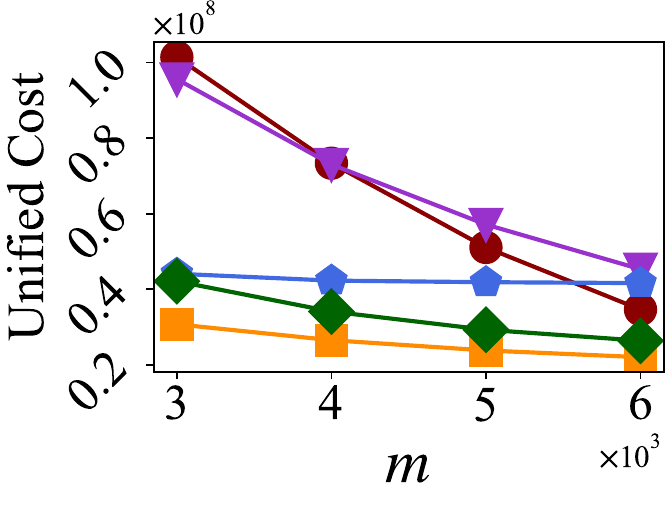}}
		\label{subfig:workers_chengdu_cost_t}}
	\hfill
	\subfigure[][{\scriptsize Unified Cost(XIA)}]{
		\scalebox{0.22}[0.22]{\includegraphics{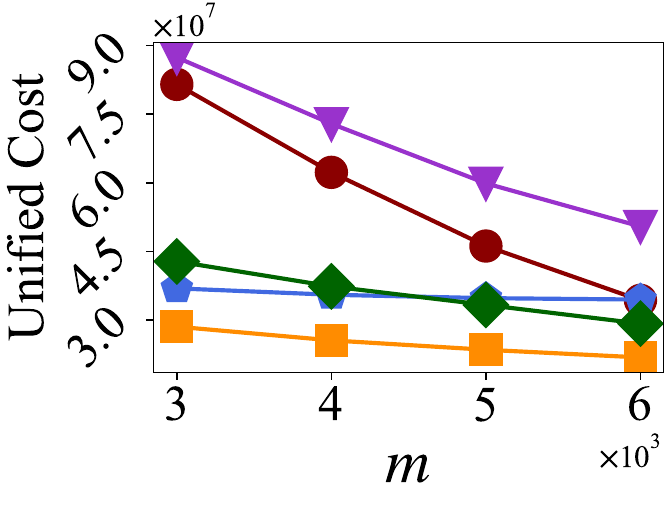}}
		\label{subfig:workers_xian_cost_t}}\vspace{-2ex}
	
	\subfigure[][{\scriptsize Service Rate(NYC)}]{
		\scalebox{0.22}[0.22]{\includegraphics{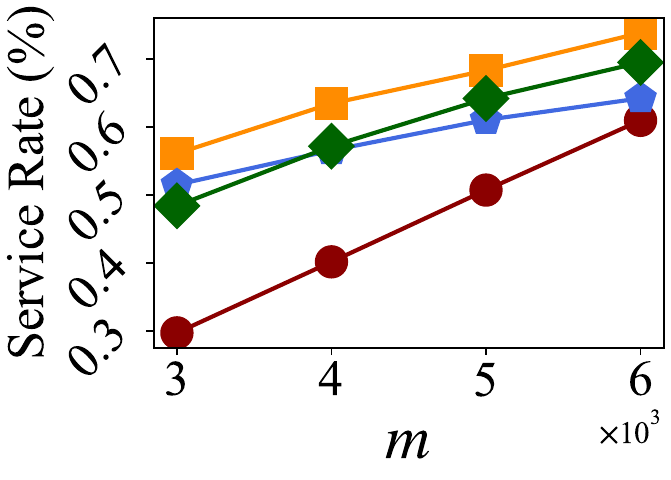}}
		\label{subfig:workers_test_served_rate}}
	\hfill
	\subfigure[][{\scriptsize Service Rate(CDC)}]{
		\scalebox{0.22}[0.22]{\includegraphics{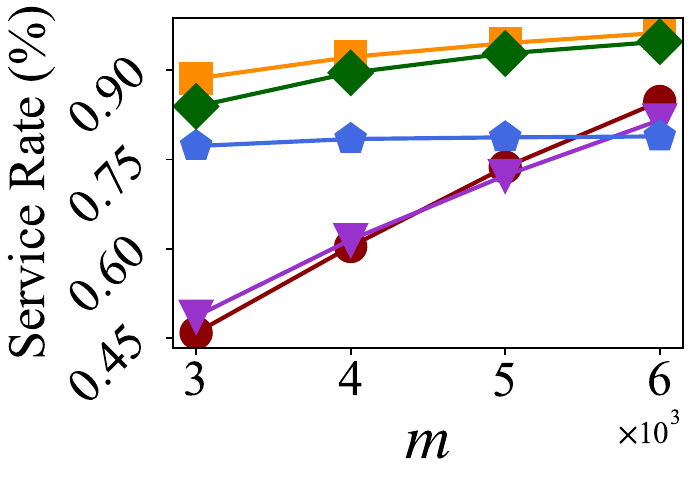}}
		\label{subfig:workers_chengdu_served_rate}}
	\hfill
	\subfigure[][{\scriptsize Service Rate(XIA)}]{
		\scalebox{0.22}[0.22]{\includegraphics{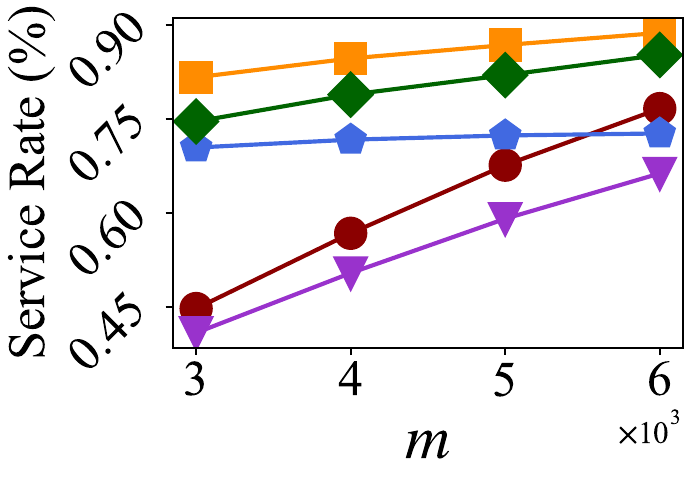}}
		\label{subfig:workers_xian_served_rate}}\vspace{-2ex}
	
	\subfigure[][{\scriptsize Running Time(NYC)}]{
		\scalebox{0.22}[0.22]{\includegraphics{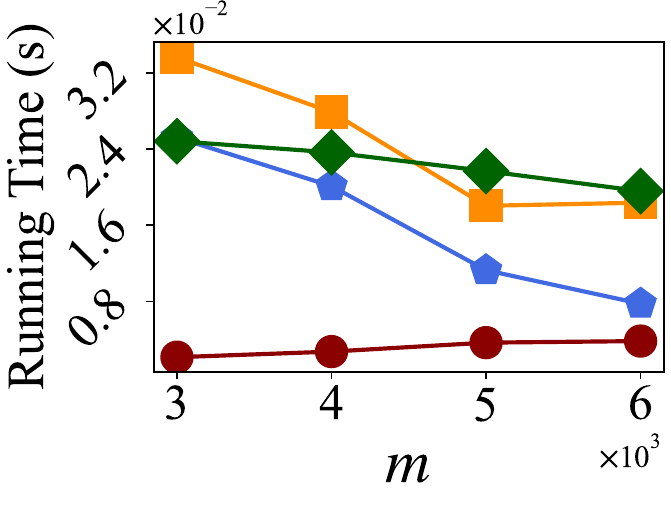}}
		\label{subfig:workers_test_running_time}}
	\hfill
	\subfigure[][{\scriptsize Running Time(CDC)}]{
		\scalebox{0.22}[0.22]{\includegraphics{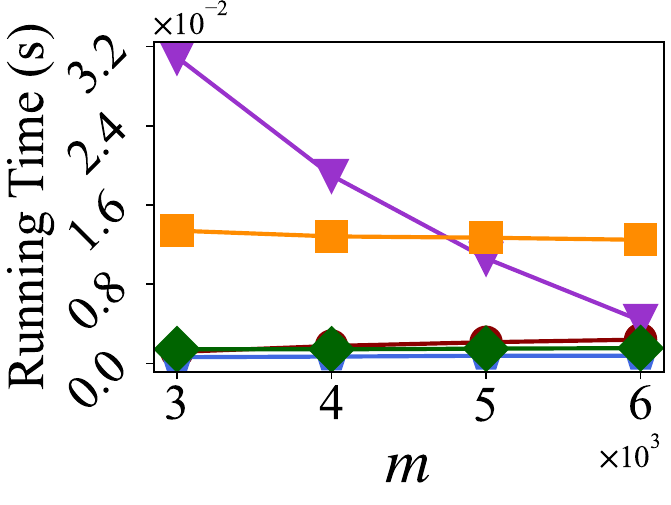}}
		\label{subfig:workers_chengdu_running_time}}
	\hfill
	\subfigure[][{\scriptsize Running Time(XIA)}]{
		\scalebox{0.22}[0.22]{\includegraphics{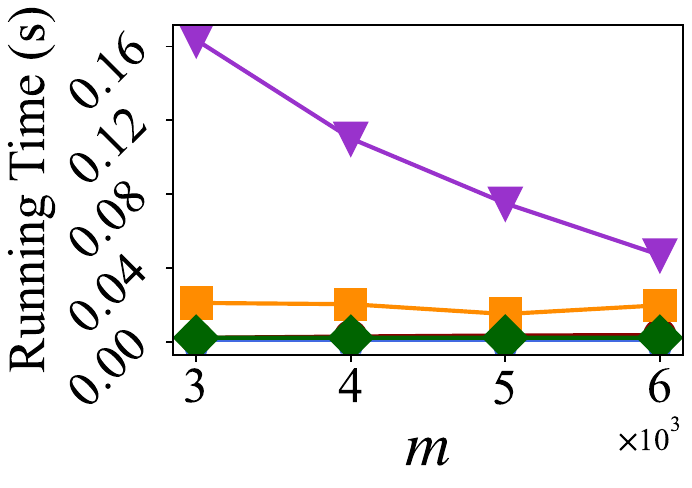}}
		\label{subfig:workers_xian_running_time}}
	\caption{\small Performance of varying $m$.}\figureBelowMargin
	\label{fig:var_worker}
\end{figure}
\noindent \textit{Impact of Varying Number of Workers. } Figure \ref{fig:var_worker} presents the results of varying the number of workers. Across all datasets, our algorithm delivers the best performance in all metrics except for running time. As the number of workers increases, both the extra time and the unified cost of all tested algorithms decrease. This is because it becomes easier to find an available worker closer to the order group, reducing the worker's response time and detour and increasing the service rate. For instance, in the NYC dataset, when $m=6000$, WATTER-expect outperformed WATTER-timeout, WATTER-online, and GDP, achieving a 4.3\%, 9.6\%, and 12.8\% improvement in service rate, respectively. Notably, the performance of the WATTER-online method exhibits less variation with increasing drivers on the CDC and XIA datasets. This is attributed to the fact that the orders in these two datasets have more dispersed pick-up and drop-off locations compared to the NYC dataset, where most orders are concentrated in the Manhattan area. As a result, the main limitation of WATTER-online is the difficulty in finding suitable shareable group.

\begin{figure}[t!]
	\subfigure{
		\scalebox{0.22}[0.21]{\includegraphics{legend.eps}}}\hfill\\\vspace{-4ex}
	\addtocounter{subfigure}{-1}
	
	\subfigure[][{\scriptsize Extra Time(NYC)}]{
		\scalebox{0.22}[0.22]{\includegraphics{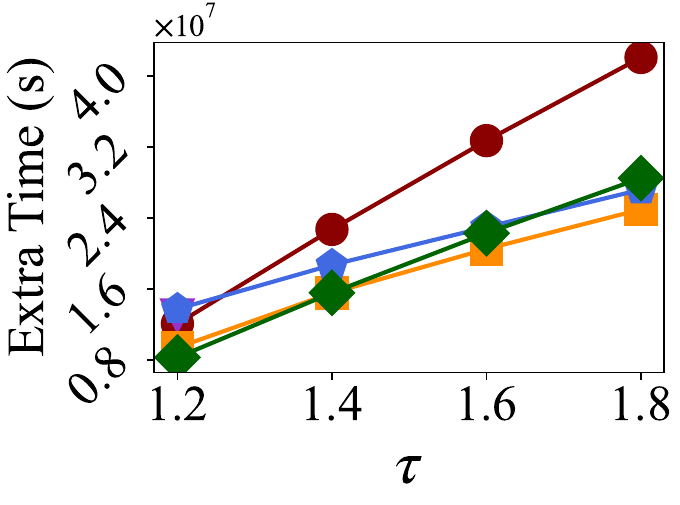}}
		\label{subfig:ddl_test_extra}}
	\hfill
	\subfigure[][{\scriptsize Extra Time(CDC)}]{
		\scalebox{0.22}[0.22]{\includegraphics{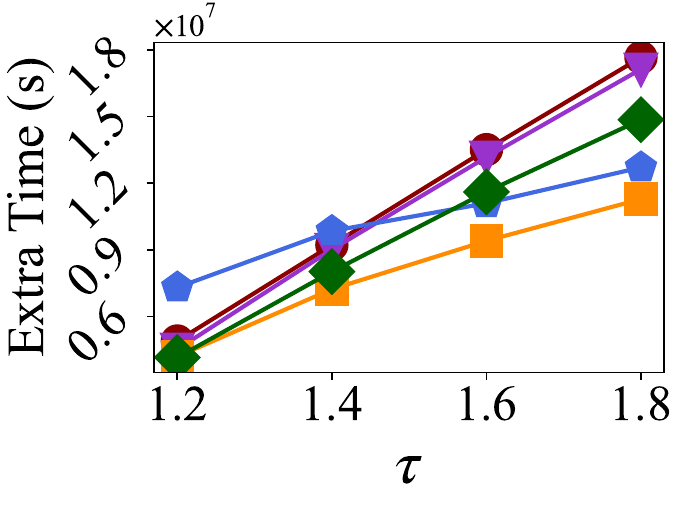}}
		\label{subfig:ddl_chengdu_extra}}
	\hfill
	\subfigure[][{\scriptsize Extra Time(XIA)}]{
		\scalebox{0.22}[0.22]{\includegraphics{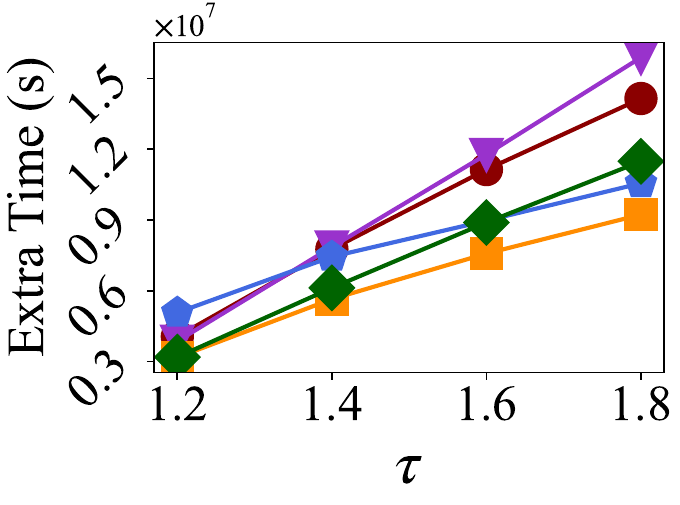}}
		\label{subfig:ddl_xian_extra}}\vspace{-2ex}
	
	\subfigure[][{\scriptsize Unified Cost(NYC)}]{
		\scalebox{0.22}[0.22]{\includegraphics{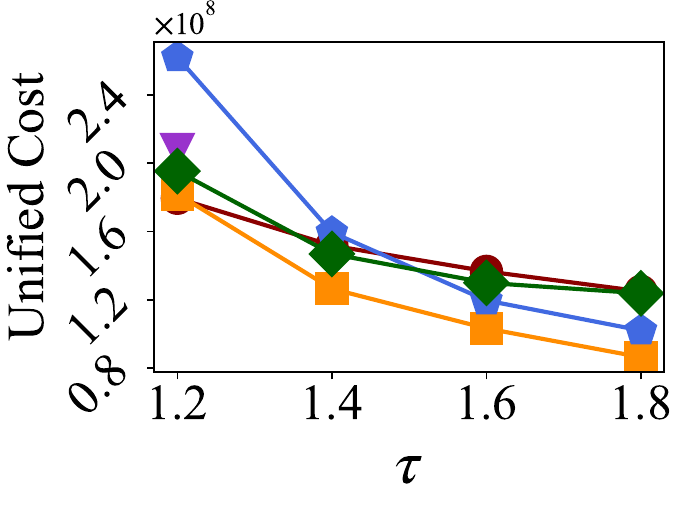}}
		\label{subfig:ddl_test_cost_t}}
	\hfill
	\subfigure[][{\scriptsize Unified Cost(CDC)}]{
		\scalebox{0.22}[0.22]{\includegraphics{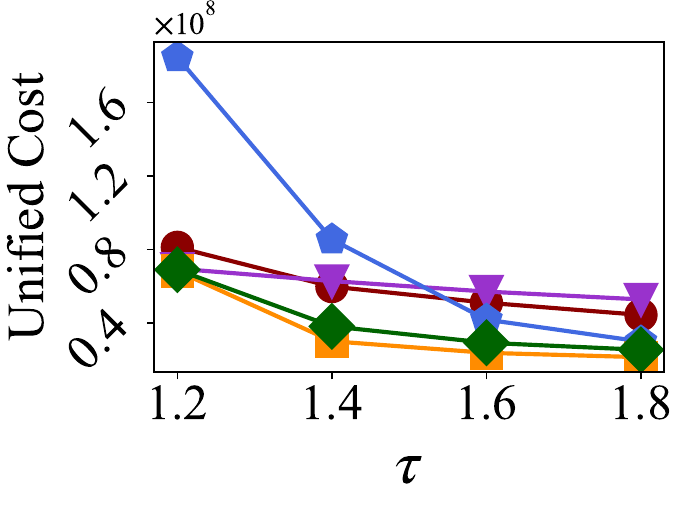}}
		\label{subfig:ddl_chengdu_cost_t}}
	\hfill
	\subfigure[][{\scriptsize Unified Cost(XIA)}]{
		\scalebox{0.22}[0.22]{\includegraphics{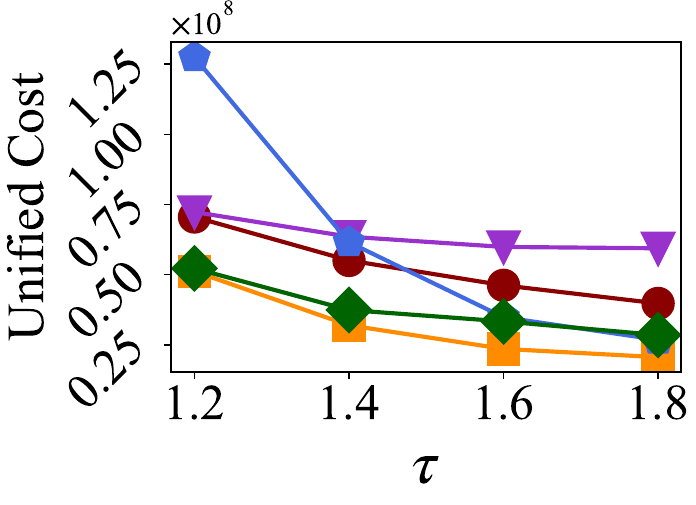}}
		\label{subfig:ddl_xian_cost_t}}\vspace{-2ex}
	
	\subfigure[][{\scriptsize Service Rate(NYC)}]{
		\scalebox{0.22}[0.22]{\includegraphics{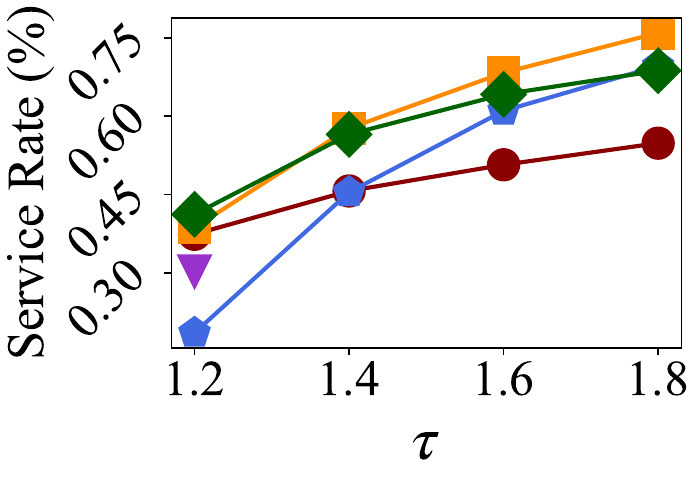}}
		\label{subfig:ddl_test_served_rate}}
	\hfill
	\subfigure[][{\scriptsize Service Rate(CDC)}]{
		\scalebox{0.22}[0.22]{\includegraphics{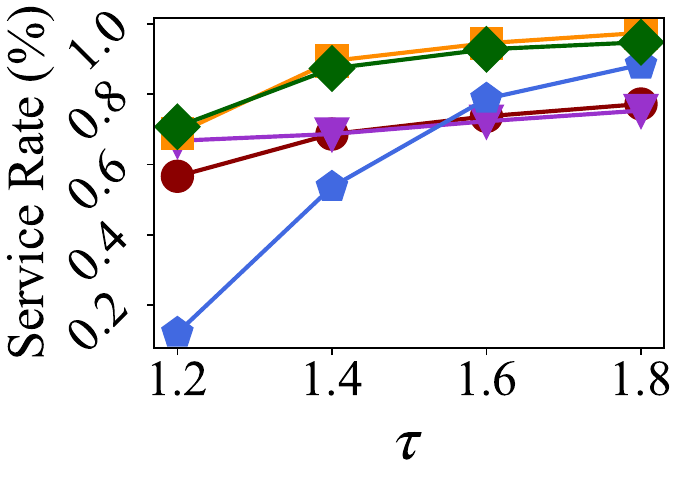}}
		\label{subfig:ddl_chengdu_served_rate}}
	\hfill
	\subfigure[][{\scriptsize Service Rate(XIA)}]{
		\scalebox{0.22}[0.22]{\includegraphics{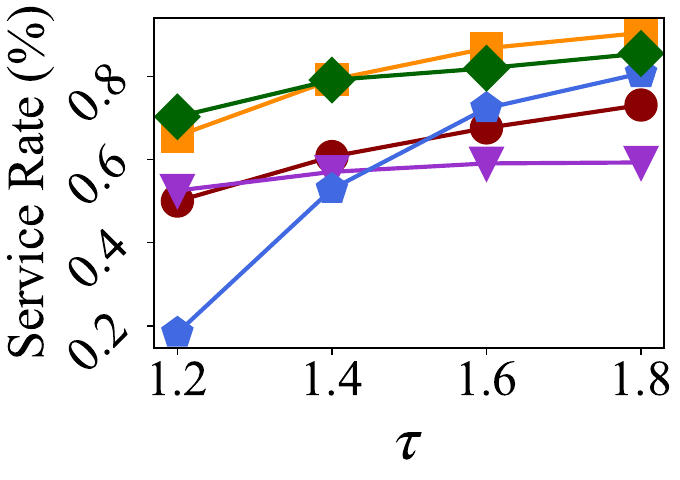}}
		\label{subfig:ddl_xian_served_rate}}\vspace{-2ex}
	
	\subfigure[][{\scriptsize Running Time(NYC)}]{
		\scalebox{0.22}[0.22]{\includegraphics{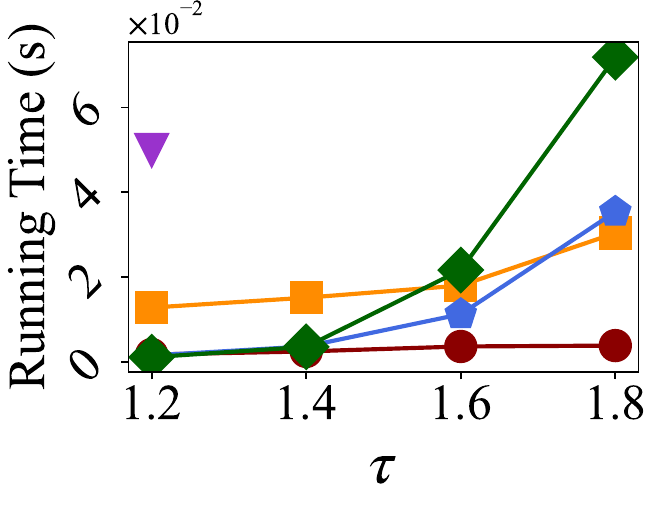}}
		\label{subfig:ddl_test_running_time}}
	\hfill
	\subfigure[][{\scriptsize Running Time(CDC)}]{
		\scalebox{0.22}[0.22]{\includegraphics{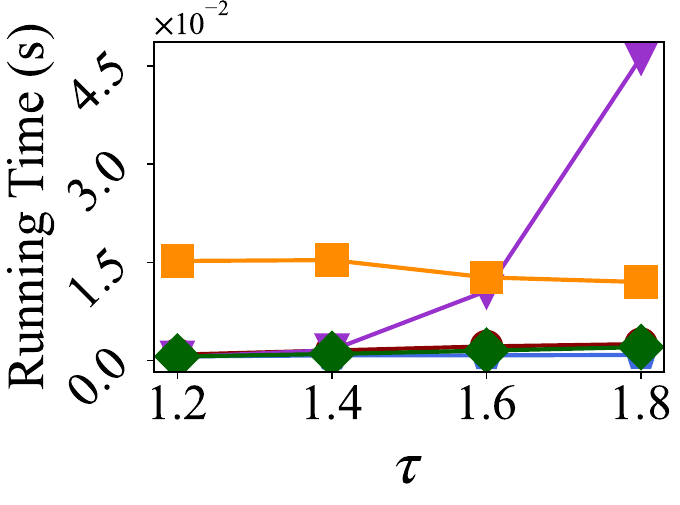}}
		\label{subfig:ddl_chengdu_running_time}}
	\hfill
	\subfigure[][{\scriptsize Running Time(XIA)}]{
		\scalebox{0.22}[0.22]{\includegraphics{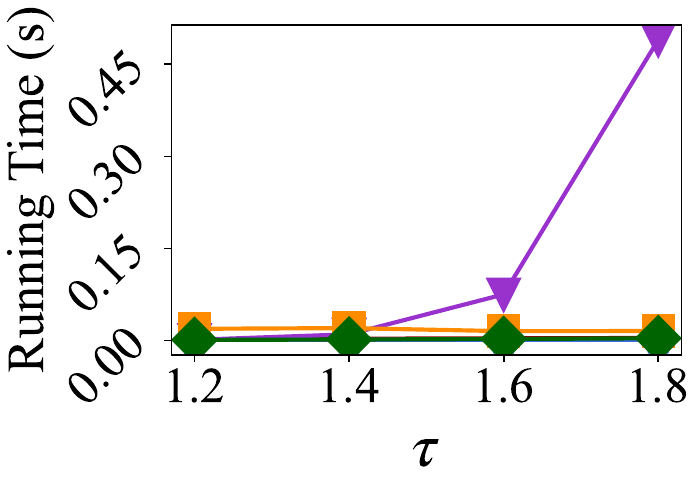}}
		\label{subfig:ddl_xian_running_time}}
	\caption{\small Performance of varying $\tau$.}\figureBelowMargin
	\label{fig:var_ddl}
\end{figure}

\revision{
	\noindent \textit{Impact of Varying Deadline.} Figure \ref{fig:var_ddl} presents the results of varying the orders' deadlines. Under small deadlines, the WATTER algorithms proposed in this paper show little difference from GDP and GAS in terms of extra time. It's because small deadlines do not allow orders to wait for too long. However, as the deadline increases, the WATTER-expect outperforms GDP and GAS. For example, when $\tau = 1.8$, WATTER-expect has a decrease of 23.1\%, 27.7\%, 48.2\%, and 65.3\% on unified cost compared to the other four algorithms on the XIA dataset, respectively.
	
	The longer deadlines improve the service rate of GDP by making it easier to insert orders into workers' routes. However, GDP fails to provide better groups due to the lack of utilization of future opportunities. For GAS, the unified cost is the highest in most cases because it only considers the payment of orders when selecting groups, without taking into account the workers' cost. However, longer deadlines increase the possibility of group orders. Furthermore, the longer wait time for riders allows more workers to become available, enhancing effectiveness of our algorithms. We observed that WATTER-online shows the most significant improvements across all performance metrics when varying the deadline. It's because the bottleneck of WATTER-online lies in the challenge of immediately finding shareable groups. Increasing the deadline greatly reduces the difficulty. }

\begin{figure}[t!]
	\subfigure{
		\scalebox{0.22}[0.21]{\includegraphics{legend.eps}}}\hfill\\\vspace{-4ex}
	\addtocounter{subfigure}{-1}
	
	\subfigure[][{\scriptsize Extra Time(NYC)}]{
		\scalebox{0.22}[0.22]{\includegraphics{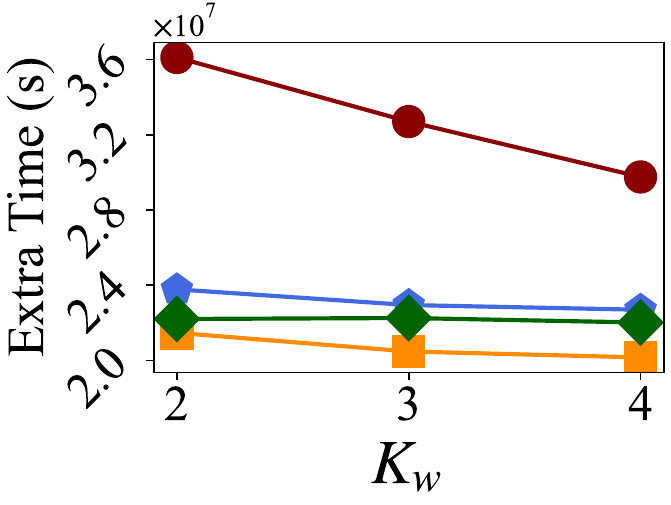}}
		\label{subfig:capacity_test_extra}}
	\hfill
	\subfigure[][{\scriptsize Extra Time(CDC)}]{
		\scalebox{0.22}[0.22]{\includegraphics{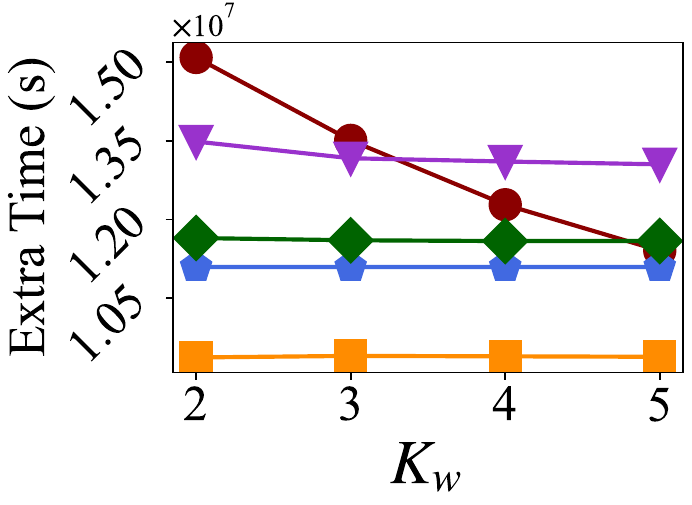}}
		\label{subfig:capacity_chengdu_extra}}
	\hfill
	\subfigure[][{\scriptsize Extra Time(XIA)}]{
		\scalebox{0.22}[0.22]{\includegraphics{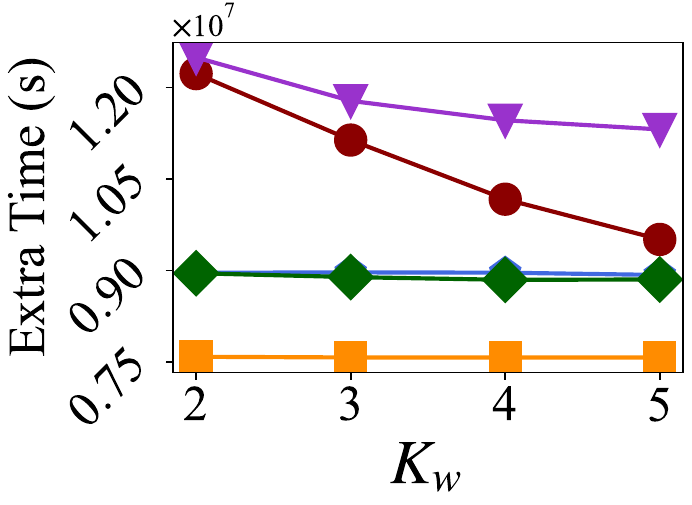}}
		\label{subfig:capacity_xian_extra}}\vspace{-2ex}
	
	\subfigure[][{\scriptsize Unified Cost(NYC)}]{
		\scalebox{0.22}[0.22]{\includegraphics{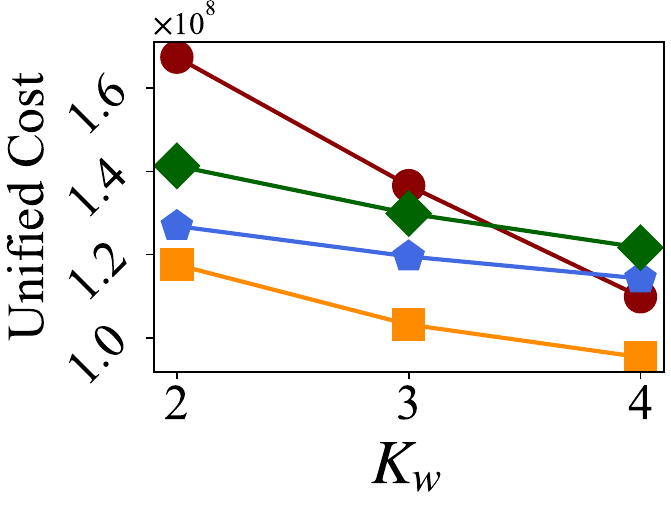}}
		\label{subfig:capacity_test_cost_t}}
	\hfill
	\subfigure[][{\scriptsize Unified Cost(CDC)}]{
		\scalebox{0.22}[0.22]{\includegraphics{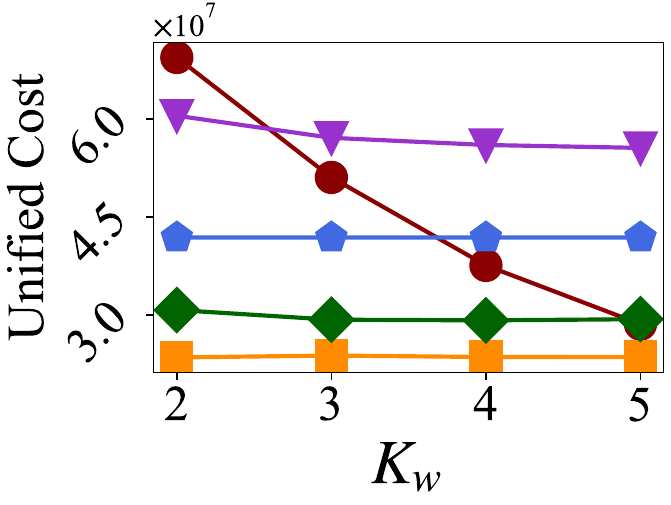}}
		\label{subfig:capacity_chengdu_cost_t}}
	\hfill
	\subfigure[][{\scriptsize Unified Cost(XIA)}]{
		\scalebox{0.22}[0.22]{\includegraphics{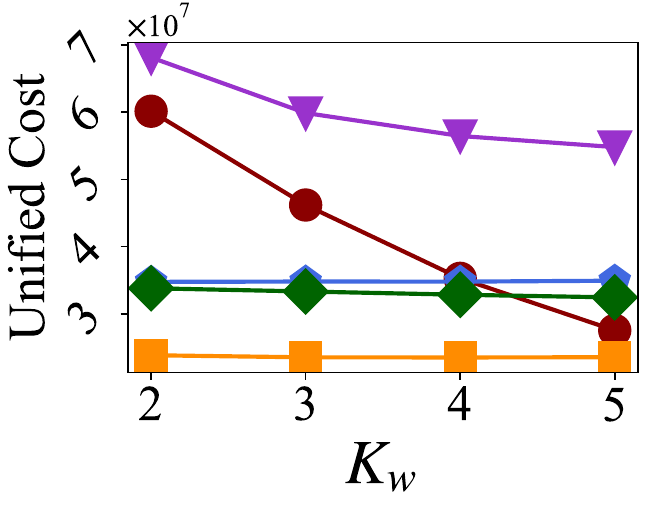}}
		\label{subfig:capacity_xian_cost_t}}\vspace{-2ex}
	
	\subfigure[][{\scriptsize Service Rate(NYC)}]{
		\scalebox{0.22}[0.22]{\includegraphics{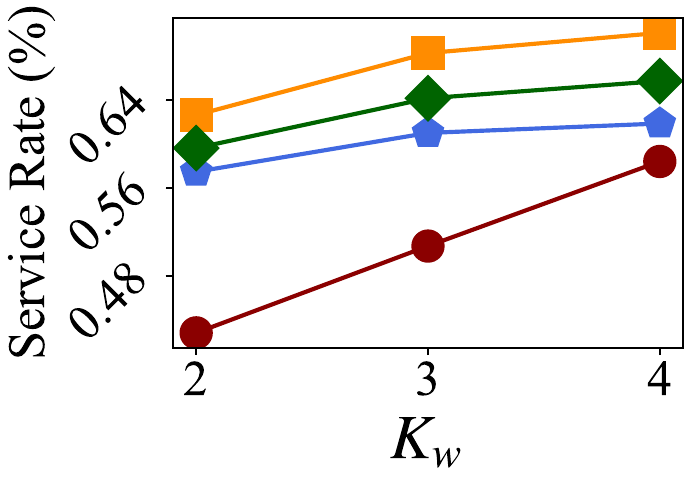}}
		\label{subfig:capacity_test_served_rate}}
	\hfill
	\subfigure[][{\scriptsize Service Rate(CDC)}]{
		\scalebox{0.22}[0.22]{\includegraphics{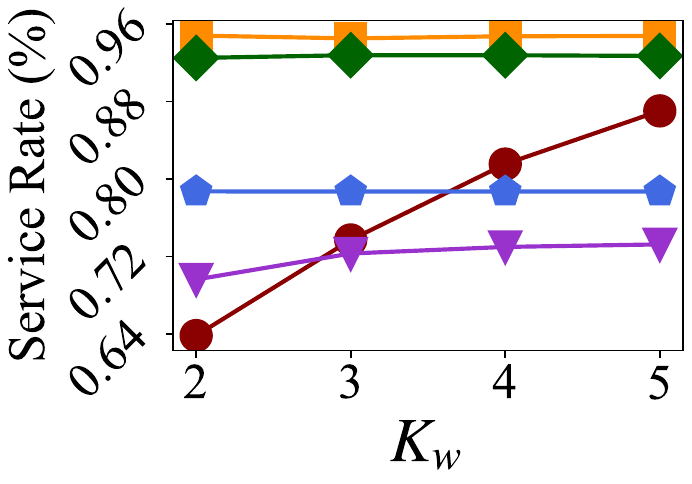}}
		\label{subfig:capacity_chengdu_served_rate}}
	\hfill
	\subfigure[][{\scriptsize Service Rate(XIA)}]{
		\scalebox{0.22}[0.22]{\includegraphics{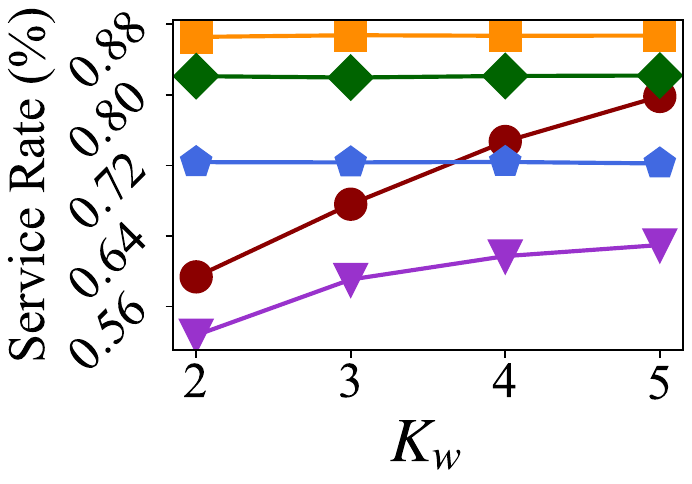}}
		\label{subfig:capacity_xian_served_rate}}\vspace{-2ex}
	
	\subfigure[][{\scriptsize Running Time(NYC)}]{
		\scalebox{0.22}[0.22]{\includegraphics{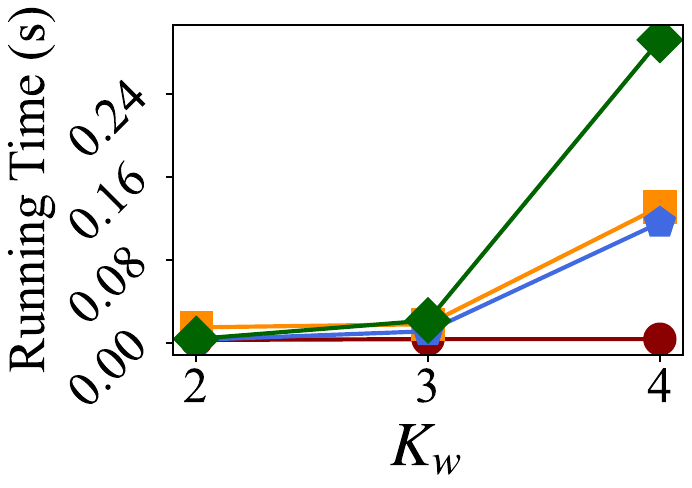}}
		\label{subfig:capacity_test_running_time}}
	\hfill
	\subfigure[][{\scriptsize Running Time(CDC)}]{
		\scalebox{0.22}[0.22]{\includegraphics{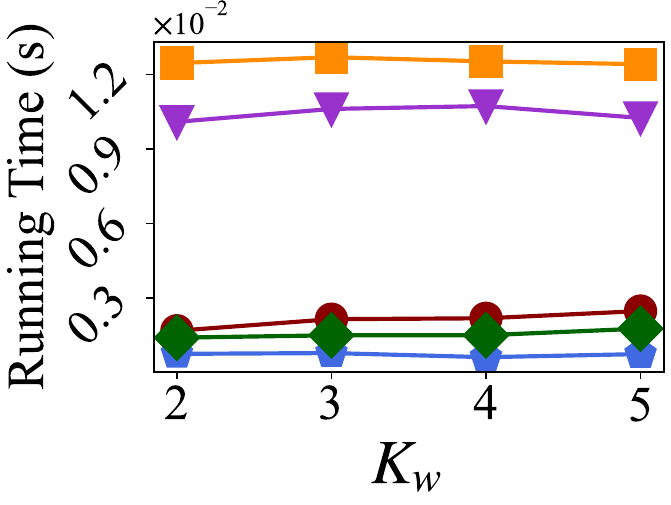}}
		\label{subfig:capacity_chengdu_running_time}}
	\hfill
	\subfigure[][{\scriptsize Running Time(XIA)}]{
		\scalebox{0.22}[0.22]{\includegraphics{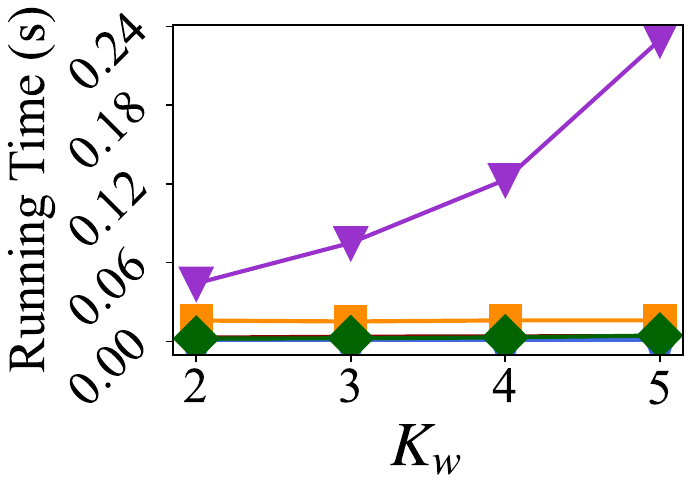}}
		\label{subfig:capacity_xian_running_time}}
	\caption{\small Performance of varying $K_w$.}\figureBelowMargin
	\label{fig:var_capacity}
\end{figure}

\noindent \textit{Impact of Varying Maximum Capacity of Workers.} Figure \ref{fig:var_capacity} presents the results of varying the workers' vehicles maximum capacity. The WATTER-expect algorithm shows superiority when the maximum capacity changes. For instance, when $K_w = 4$, WATTER-expect achieved the following reductions in extra time compared to the baselines: 15.3\%, 18.9\%, 23.5\%, and 28.1\% on CDC dataset, respectively. Increasing the maximum capacity $K_w$ significantly decreases the extra time and unified cost for GDP. However, our algorithms show less fluctuation due to four reasons: (1) when optimizing for extra time, the performance of the group with a capacity of 2 is dominant; (2) the GDP greedily inserts workers without considering waiting time and dispatch the order as long as there is a slot available for insertion. Increasing the capacity limit can effectively improve the service rate, further optimizing extra time and unified cost; (3) the non-preemptive mode we adopt means that workers cannot serve two groups at the same time. If the capacity of a worker's vehicle exceeds the size of the group to be served, the excess cannot be utilized efficiently; (4) the shareability graph on the CDC and XIA datasets is sparse, making it difficult to enumerate groups whose size is greater than 2. This can also be verified by the running time, as our algorithms show an increase in running time as the maximum capacity increases in the NYC dataset, but not on the CDC and XIA datasets.

\noindent \textit{Trade-off between Response and Detour Time.} 
We present the average proportions of response time and detour time for all served orders in Figure \ref{fig:time_compare_bin}. The results are obtained on the CDC dataset with default parameters.  We can observe that among the three variants of the WATTER algorithm, the WATTER-online achieves the lowest response time but also results in the highest detour time. On the other hand, the WATTER-timeout exhibits the highest response time and the lowest detour time. This supports our assumptions regarding response time and detour time: the longer waiting increases the likelihood of finding a higher-quality group, leading to a smaller average detour time. Meanwhile, the GDP and GAS show the same trend with WATTER-online. In contrast, WATTER-expect achieves the lowest extra time through effectively balancing response time and detour time, which allows orders to wait for an appropriate short time, and effectively reduces detour time. 

\begin{figure}[t!]\centering \vspace{-3ex}
	\scalebox{0.2}[0.2]{\includegraphics{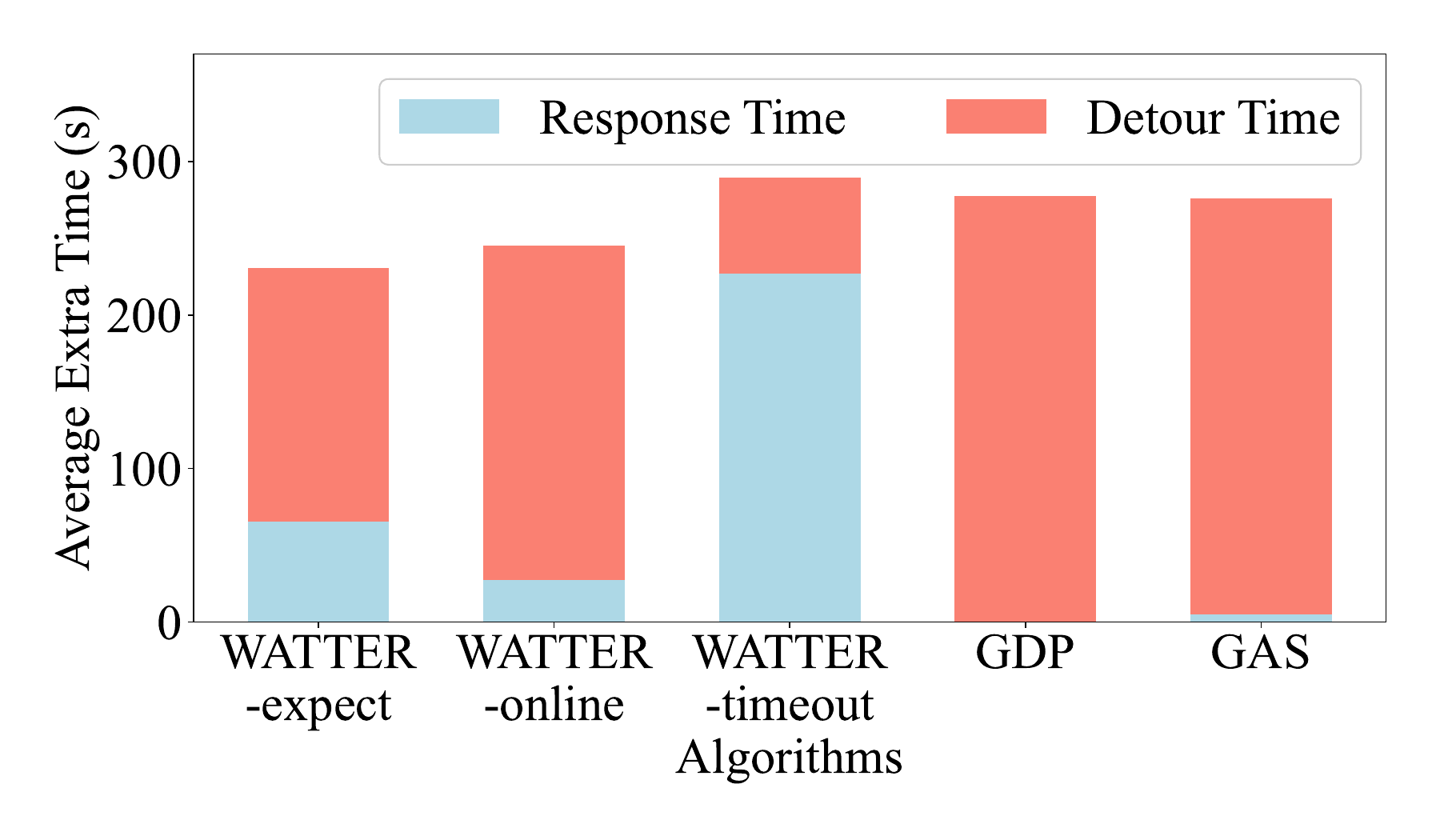}}
	\caption{\small Trade-off between response and detour time.}
	\label{fig:time_compare_bin}
\end{figure}

\noindent \textit{Summary of the Experimental Results}
We summarize the major experimental results in the following.

\begin{enumerate}
	\item  In terms of the total extra time and unified cost, our order pooling management algorithms WATTER-expect, WATTER-online and WATTER-timeout show superiority compared to the online-based method GDP and batch-based method GAS. 
	\item The threshold-based strategy has a noticeable effect on balancing the response time and group quality. The WATTER-expect have less extra time than WATTER-online and WATTER-timeout up to 15.3\% and 18.9\% on CDC dataset when $K_w=4$. 
	\item GDP are more efficient than our algorithms when the orders are more similar(e.g., in the NYC dataset). It is because the shareability graph are larger and denser, which makes huge group enumeration cost.
\end{enumerate}

\section{Related Work}
\label{sec:related}

The ridesharing problem is a variant of the dial-a-ride problem, which involves planning the routes of workers to serve orders with specific pick-up and drop-off locations. Unlike the dial-a-ride problem ~\cite{cheng2017utility,cheng2019queue, chen2019minmax}, the ridesharing problem primarily focuses on the additional benefits of grouping orders together. Recent studies ~\cite{tong2018prunegdp, wang2020demand, zeng2020gas} can be classified mainly based on processing frameworks and optimization objectives.

\noindent \textit{Process frameworks.} Online-based methods use the insertion operator as a heuristic solution for planning work routes. They greedily and incrementally insert each newly arriving order into the worker's current route based on certain objectives. Zheng \textit{et al.}~\cite{zheng2013tshare} enumerate all possible insertion positions to search for the optimal solution. Huang \textit{et al.}~\cite{huang2014kinetic} proposed the structure called kinetic tree, which maintains and provides the optimal schedule for vehicles when new orders arrive. Tong \textit{et al.}~\cite{tong2018prunegdp} proposed a dynamic programming (DP) algorithm to insert and check constraints, reducing the operation time from cubic to linear. Wang \textit{et al.}~\cite{wang2020demand} proposed demand-aware insertion based on predictions about future orders.

Batch-based methods process all orders within a mini-batch time at once. They typically enumerate all possible groups to find the optimal one and then assign the groups to workers \cite{zheng2018order, bei2018bimatch, cheng2019queue, zeng2020gas, zheng2019auction}. Bei \textit{et al.}~\cite{bei2018bimatch} formulated the problem of combinatorial optimization between orders and workers. Cheng \textit{et al.}~\cite{cheng2019queue} use machine learning models to predict future vehicle demand. They then used a queue-theoretic framework to balance the demand and supply. Zeng \textit{et al.}~\cite{zeng2020gas} proposed an index called additive tree to accelerate the enumeration and greedily choose the most profitable group to serve. The effectiveness and runtime of batch-based methods depend on the batch size setting. If the batch size is large, the runtime increases exponentially. In contrast, the utility achieved is less competitive than that of online-based methods.

\noindent \textit{Optimization objectives.} The primary optimization objectives for ridesharing services are from the perspectives of the platform and workers, with other objectives acting as constraints. From the platform's perspective, previous studies~\cite{cheng2019queue, zeng2020gas, wang2022revenue} have mainly focused on maximizing revenue and the number of served riders. The platform's revenue is calculated by subtracting travel cost of workers from the payment of riders. Given a fixed service rate, reducing the worker's travel cost can increase revenue.

Many studies~\cite{xu2019insertion, liu2020mobility, Haliem2021modelfree,li2020consensus} focus on decreasing travel costs from the worker's perspective. This can increase both the platform's revenue and the workers' earnings. To unify these objectives, Tong \textit{et al.}~\cite{tong2018prunegdp} introduced the objective of a unified cost, which integrates the three goals into one and illustrates the relationships among them.

Regarding the rider's perspective, most existing studies~\cite{cheng2017utility, chen2018price, tong2018prunegdp} impose the time constraint that riders must be delivered before the deadline. Among these studies, the manually set deadline is an important parameter that affects the algorithms' performance. Flow time~\cite{Firat2011flow1, HAME2011flow2, xu2019insertion} is a commonly used objective for improving the rider satisfaction.

However, due to limitations in the processing framework, studies on ridesharing provide group results to riders either immediately or in mini-batches, without taking into account the response time. In this paper, we propose a new optimization objective for the METRS problem that balances the response time and detour time to improve riders' satisfaction. Different from existing processing frameworks, we propose a novel order pooling management algorithm to handle the balance problem, which is more effective and efficient than previous studies.

\section{Conclusion}
\label{sec:conclusion}
In this paper, we study the Minimal Extra Time RideSharing (METRS) problem in which orders arrive dynamically, and the platform needs to group them and assign workers to serve the groups as many as possible while maximizing rider satisfaction. We prove that the METRS problem is NP-hard. To address this challenge, we propose an efficient framework WATTER. It utilizes order pooling management algorithm to allow riders' orders to wait in a temporal shareability graph-based order pool until they can be grouped effectively. We also devise an average extra time threshold-based grouping strategy that determines when to dispatch the orders in the pool. To adapt to different spatiotemporal environments of orders, we model the decision process as a Markov Decision Process (MDP) and use a reinforcement learning model to solve it. Through extensive experiments on three real dataset, we demonstrate the efficiency and effectiveness of our WATTER framework showing that it can outperform two SOTA baselines.

\bibliographystyle{ieeetr}
\bibliography{add}
\newpage
\appendix

\subsection{Order Arrival Algorithm}

We illustrate the update triggered by new order arrival, as shown in Algorithm \ref{algo:order_arrival}. To accelerate the subsequent decision-making process, we maintain a map $G_b$ that stores the current best group information of orders. When a new order arrives, we take the following three steps: (1) Find neighbors and add edges (line 1); (2) Enumerate new $k$-cliques (line 2); (3) Update the map $G_b$ based on the enumeration results (lines 3-11). In the third step, we utilize the Insertion Algorithm~\cite{tong2018prunegdp} to generate the fastest feasible route for each group. Although the neighbors may have better groups due to the arrival of new orders, but the entire insertion process can still be completed with only one enumeration.

\begin{algorithm}[ht!]
	\DontPrintSemicolon
	\KwIn{order $o^{(i)}$, original graph $\mathcal{G}$, map $G_b$}
	\KwOut{updated graph $\mathcal{G}$, updated map $G_b$}
	traverse the graph $\mathcal{G}$, find neighbors that can be shared with $o^{(i)}$ and add new edges \;
	$G^{(i)}\rightarrow $ enumerate all $k$-cliques that contains $o^{(i)}$ \;
	add entry $(i, (NULL, INF))$ into the map $G_b$ \;
	
	\ForEach{clique $g \in G^{(i)}$} {
		$L \rightarrow $ fastest feasible route of $g$ \;
		\lIf {$L$ doesn't exist} {\textbf{continue}}
		$t_e \rightarrow$ the average extra time of requests in $L$\;
		\ForEach{$o^{(j)} \in g$}{
			$t_e^{(j)} \rightarrow$ the current optimal average extra time stored in $G_b[j]$ \;
			\If{$t_e < t_e^{(j)}$}{ 
				replace $G_b[j]$ with $(j, (g, t_e))$ \;
			}
		}
	}
	
	\Return{updated graph $\mathcal{G}$, updated map $G_b$} \;
	\caption{Order Arrival}
	\label{algo:order_arrival}
\end{algorithm}

\subsection{Order Departure and Expiration Algorithm}

In addition to new orders, updates to the graph can also be triggered by dispatch or rejection of orders, as well as expiration of groups or edges. We can process these updates in one subroutine, as shown in Algorithm \ref{algo:order_departure}. If a leave order is included in the current best groups of other orders, then the current best group of these orders needs to be recomputed. We use set $C$ to store candidate orders to reduce duplicate calculations (lines 1-6).  As for edge and group expiration, the process is similar to departure. We can use the group with a size of two to present the edge. Then, we only need to consider orders whose best group contains the orders in the group (lines 7-8). For each candidate order, we reuse and make a little  Algorithm \ref{algo:order_arrival} with a small modification to generate the current best group (lines 9-10). The only different of the modification is there's no need to update the graph structure of the candidate orders.

\subsection{Detailed Training Algorithm of DQN}

Given historical data, we simulate the dispatch process using strategy Equation \ref{eq:strategy} in the order pool and collect experience transitions $(s_t, a, $$ s_{t + \Delta t}, r_t)$. Algorithm \ref{algo:simulator} presents the details of an episode for experience collection and network update. The approach is implemented using the aforementioned framework, and we present the key sub-procedures for simplicity. We maintain a replay memory $M$ to store the transitions. Additionally, due to the uncertainty of the next state for wait actions, we use a buffer $B$ to temporarily store these transitions (line 3). To avoid the sparsity of the demand distribution, we use the first $m$ orders to initialize it (line 4). Note that the update of the distribution only needs constant time (Line 7). Then, we decide the action for each order in the graph using the main network $V$ and strategy Algorithm \ref{algo:expectation_based_make_decision} (line 17-19). For the trade-off of exploration and exploitation, the action may be altered with probability $\epsilon$ (Line 20). If the order has reached the maximum waiting time, we also dispatch it (line 21). There are two kinds of situations that lead to the termination state: (1) timeout but not included in any group (line 11-15); (2) dispatched with a certain group (line 21-25). In these situations, the temporal transitions in buffer $B$ will be popped and flushed with new transitions (Algorithm \ref{algo:replace_terminate}), while wait action leads to transitions temporarily stored in buffer $B$.

\begin{algorithm}[ht!]
	\DontPrintSemicolon
	\KwIn{order group $g$, original graph $\mathcal{G}$, map $G_b$}
	\KwOut{updated graph $\mathcal{G}$, updated  map $G_b$}
	
	initialize order set $C$ that need re-compute best group \;
	\If {departure} {
		\ForEach{$o^{(i)} \in g$}{
			remove all edges related to order $o^{(i)} \in g$ \;
			$C^{(i)} \rightarrow $ the orders that have $o^{(i)}$ as a member of their best group \;
			$C \rightarrow C \cup C^{(i)}$ \;
		}
	}
	\If {expiration} {
		$C \rightarrow$ the orders that have all orders in $g$ as members of their best group \;
	}
	\ForEach{$o^{(j)} \in C$}{
		Order\_Arrival($o^{(j)}$) \;
	}
	
	\Return{updated graph $\mathcal{G}$, updated map $G_b$} \;
	\caption{Order Departure or Group Expiration}
	\label{algo:order_departure}
\end{algorithm}

\begin{algorithm}[ht!]
	\DontPrintSemicolon
	\KwIn{original state $s$, action $a$, reward $r$, buffer $B$ and memory $M$}
	\KwOut{updated buffer $B$ and memory $M$}
	
	$\delta' \leftarrow (s, a, o, TERMINATE)$ \;
	
	retrieve last transition $\delta=(s,a,o,s')$ from $B$\;
	replace $s' \in \delta$ with $s$ \;
	
	flush $\delta, \delta'$ into $M$ \;
	
	\Return{updated buffer $B$ and memory $M$} \;
	\caption{Replace Terminate}
	\label{algo:replace_terminate}
\end{algorithm}

\begin{figure}[t!]
	\subfigure[][{\scriptsize Service Rate}]{
		\scalebox{0.25}[0.25]{\includegraphics{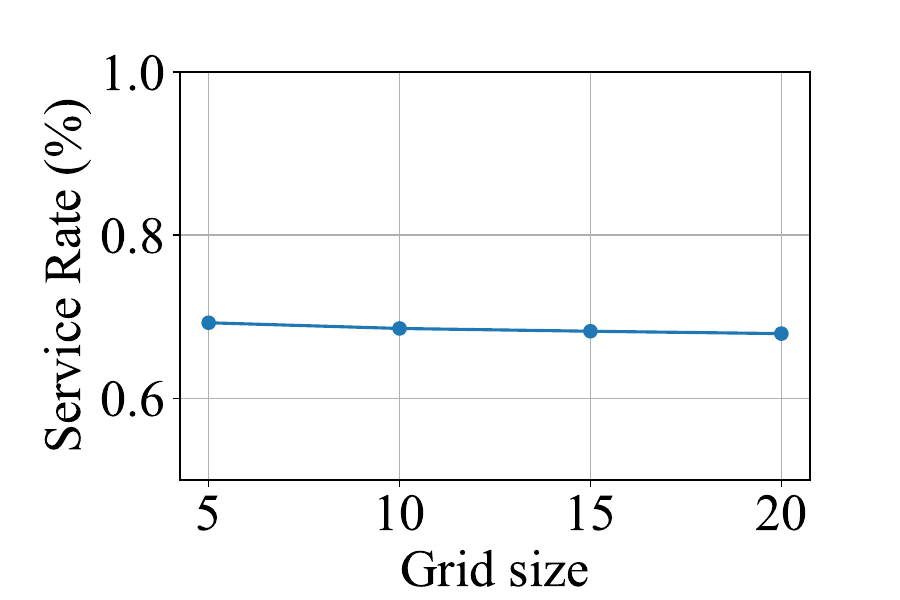}}
		\label{subfig:grid_served_rate}}
	\hfill
	\subfigure[][{\scriptsize Average Extra Time}]{
		\scalebox{0.25}[0.25]{\includegraphics{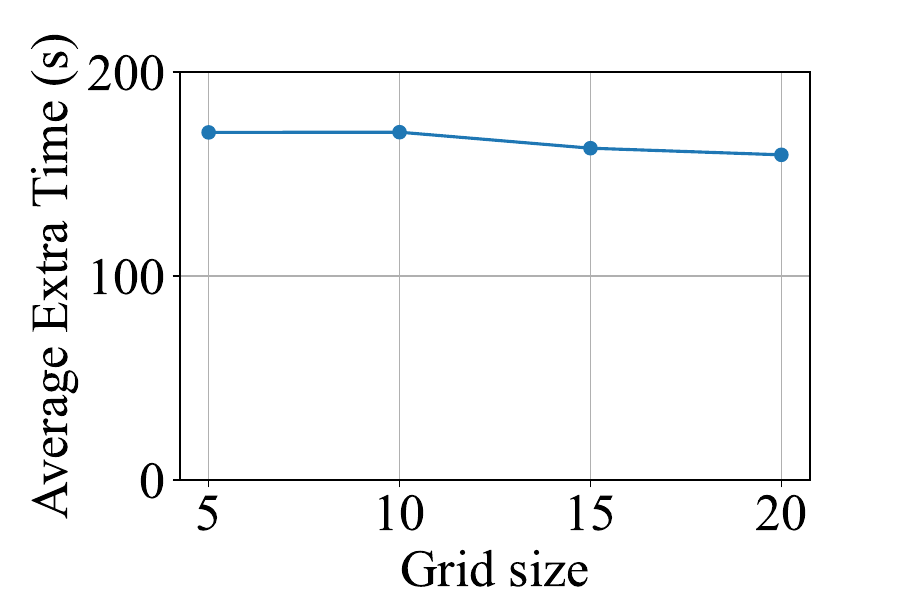}}
		\label{subfig:grid_avg_extra}}
	
	\subfigure[][{\scriptsize Unified Cost}]{
		\scalebox{0.25}[0.25]{\includegraphics{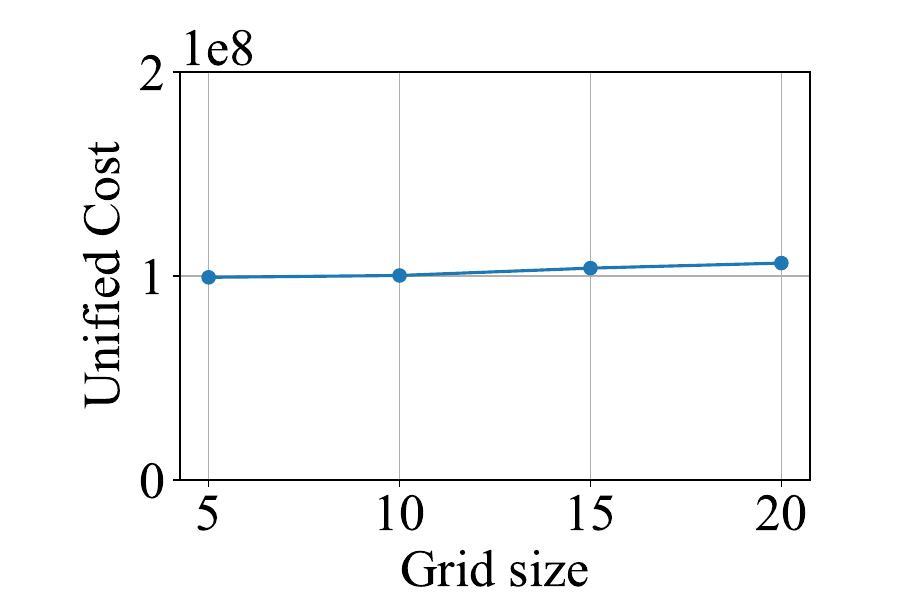}}
		\label{subfig:grid_cost}}
	\hfill
	\subfigure[][{\scriptsize Running Time}]{
		\scalebox{0.25}[0.25]{\includegraphics{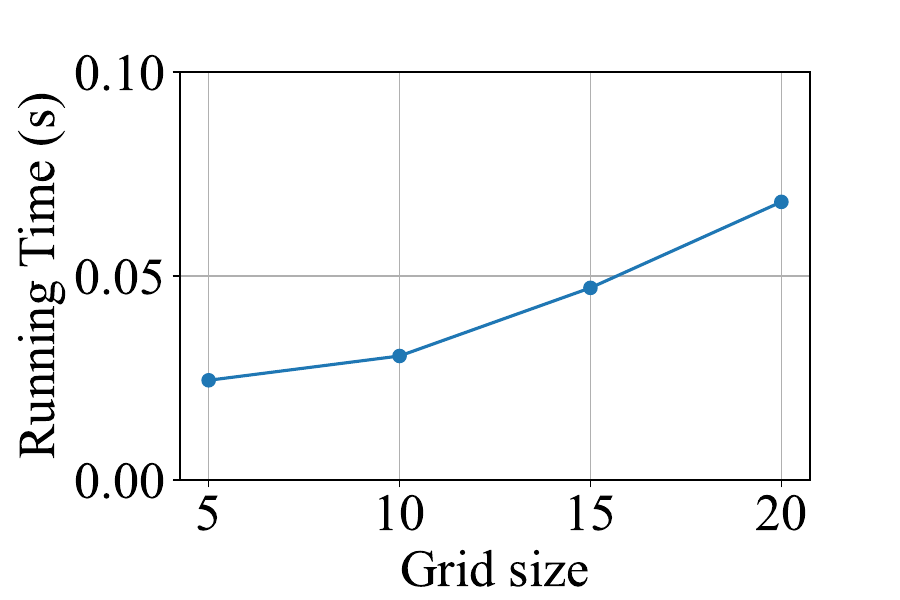}}
		\label{subfig:grid_running_time}}
	
	\caption{\small Performance of varying grid size.}\figureBelowMargin
	\label{fig:var_grid}
\end{figure}

\subsection{Impact of Varying Grid Size}

To investigate the impact of grid size on algorithm performance, we conducted tests using the WATTER-expect algorithm on the NYC dataset. We divided the entire road network into grids of varying sizes: $5\times 5$, $10\times 10$, $15\times 15$, and $20\times 20$, while keeping the other parameters at their default values. The grid size primarily affects the accuracy of the model's predictions. More grids require more detailed training in order to achieve similar levels of accuracy. As the model takes distribution information into account, the number of grids is related to the feature dimensions, and using more grids results in a significant increase in algorithm running time.  Figure \ref{fig:var_grid} presents the results of varying the number of grids. After increasing the number of grids, there is little change in the unified cost, average extra time, and service rate. However, the average running time per request significantly increases.

\begin{algorithm}[t!]
	\DontPrintSemicolon
	\KwIn{A set $W$ of $m$ workers, a set $O$ of $n$ orders ordered by arriving time}
	\KwOut{Learned state value function $V$}
	initialize shareability graph $\mathcal{G}$ \;
	initialize network $V$, $\hat{V}$ \;
	initialize replay memory $M$ and buffer $B$\;
	initialize the demand and supply distribution $O,E,D$ by the first $m$ orders \;
	\ForEach{$o^{(i)} \in O$} {
		insert $o^{(i)}$ into $\mathcal{G}$ and set $t^{(i)} \rightarrow t_s$ \;
		update demand distribution $O,E$ \;
		
		\ForEach{$o^{(j)} \in \mathcal{G}$} {
			$g$ $\leftarrow$ find a best group in $\mathcal{G}$ contains $o^{(j)}$ \;	
			\If{$g=$\texttt{NULL}} {
				\If{is\_timeout($o^{(j)}$)}{
					remove $o^{(j)}$ from $\mathcal{G}$ \;
					construct current state $s$ \;
					Replace\_Terminate($s$, 0, 0, $B$, $M$) \;
					update $O,E,D$\;
				}
				\textbf{continue} \;
			} 
			
			construct current states $S$ of orders in $g$ \;
			compute the state values $V$ by $S$ \;  
			$a\leftarrow$ threshold\_based\_make\_decision($g,V,t$) \;
			alter $a$ with another action with probability $\epsilon$\;
			
			\If{$a = 1$ \textnormal{or} is\_timeout($o^{(j)}$)} {
				\ForEach{$o^{(k)} \in g$} {
					terminate($S[k]$, a, $p^{(k)}- t_d^{(k)}$, $B$, $M$) \;
				}
				assign the $g$ to a worker to serve. \;
				update $O,E,D$\;
			} \Else {
				$\delta' \leftarrow (S[i], a, \Delta t, NULL)$ \;
				add $\delta '$ into $B$ \;
			}
		}
	}

	\Return{$M$} \;
	\caption{Deep-$Q$-Network (DQN) Learning for Estimating Value Function }
	\label{algo:simulator}
\end{algorithm}

\subsection{Result of the Convergence Curves}

\begin{figure}[h]
	\subfigure{
		\scalebox{0.3}[0.3]{\includegraphics{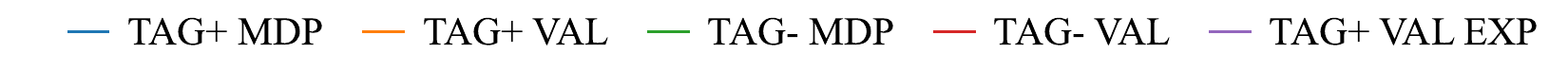}}}\hfill\vspace{-2ex}
	\addtocounter{subfigure}{-1}
	
	\subfigure[][{\scriptsize Cumulative Reward Trend}]{
		\scalebox{0.25}[0.25]{\includegraphics{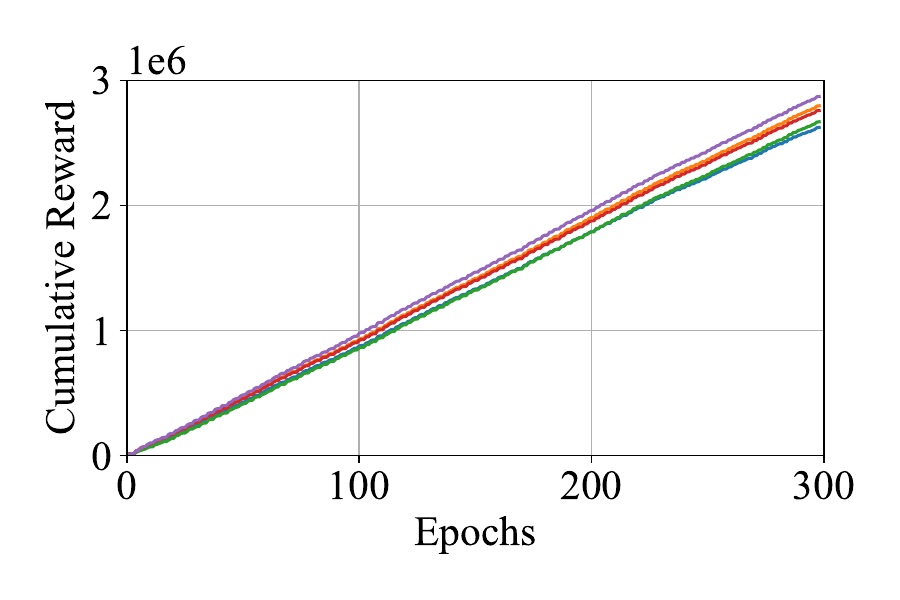}}
		\label{subfig:reward_trend}}
	\subfigure[][{\scriptsize Loss Trend}]{
		\scalebox{0.25}[0.25]{\includegraphics{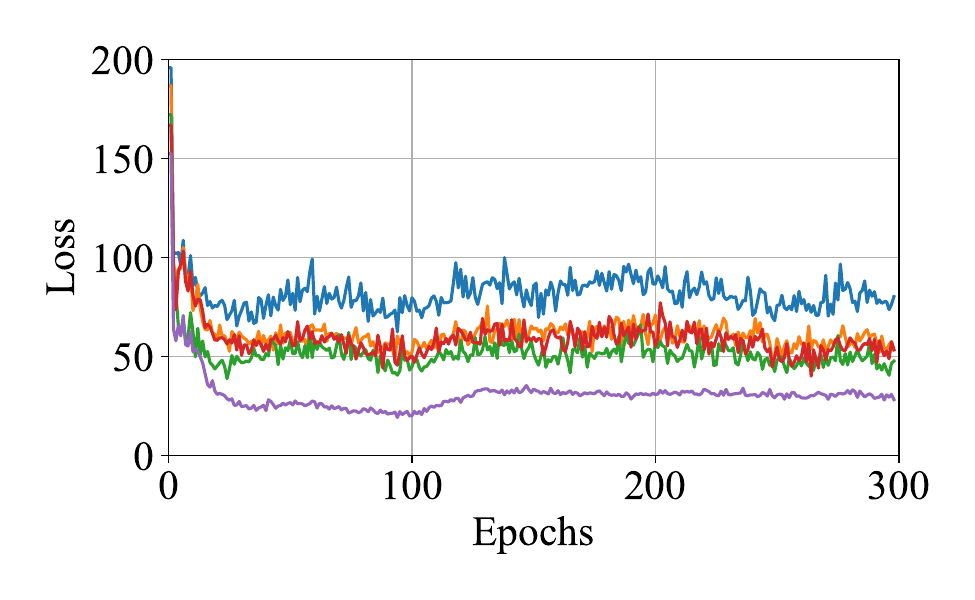}}
		\label{subfig:loss_trend}}
	\hfill
	
	\caption{\small Convergence Curves.}\figureBelowMargin
	\label{fig:convergence_curve}
\end{figure}

There can be multiple implementation approaches for the experience generation in the training process. For example, (1) the agent can make dispatch decisions based on the value function; (2) the value function can be used as an estimation of threshold. The first approach has a limitation in that the difference between the current state and the next state may not be significant. If the neural network is not deep enough, it may struggle to distinguish between the two states, leading to a policy with high randomness. Additionally, considering the delay, we do not employ large-scale neural networks as the implementation of the value function. The second approach has a drawback in the initial training phase, where parameter initialization can bias the value function towards smaller expected thresholds, making it challenging to obtain high-quality training experiences. Therefore, we adopt an off-policy approach to enhance the training process. In the initial training phase, we collect training experiences using the threshold-based strategy and the optimal expected threshold from Section \ref{sec:solution2}. In the later stage, we utilize the second approach to gather training experiences for fine-tune. 

Regarding the experiments on model convergence, we have prepared several training configurations. The labels TAG+ and TAG- indicate whether the target loss is used during training. The MDP configuration utilizes Q-learning with the value function for decision-making. In other words, the agent selects the action (dispatch or wait) that leads to the next state with higher reward as the final action. The VAL configuration employs the value function as the expected threshold for decision-making, similar to the approach described in Section \ref{sec:solution2}. Finally, the EXP label indicates that, in the first part of epochs, the GMM from Section \ref{sec:solution2} is used to generate expected thresholds, which are employed for decision-making to produce training experiences. In the later part of epochs, the value function is used as the expected threshold. In this section, there are a total of five training configurations: (1)TAG+ MDP, (2)TAG+ VAL, (3)TAG- MDP, (4)TAG- VAL, and (5)TAG+ VAL EXP. The last configuration TAG+ VAL EXP is used in WATTER-expect.

We conducted 10 repetitions of training on the NYC dataset, using default parameter configurations for each model. The resulting TD loss and reward mean values are reported to demonstrate the convergence, as depicted in Figure \ref{fig:convergence_curve}. Notably, the green curve, representing the pure Q-learning approach, exhibited the lowest cumulative reward. Conversely, the purple curve, representing the training approach used in WATTER-expect, achieved the highest cumulative reward. We observed that using the TAG+ label configuration resulted in higher loss compared to the TAG- label configuration. This is due to partial fitting of the target value, resulting in a higher mean of the value function. For the VAL approach, when utilizing the GMM to generate training experiences in the early epochs (indicated by the purple curve), the convergence effect is smoother and yielded the highest cumulative reward. However, after the 100th epoch, the loss of the purple curve showed a slight increase. This occurred because before the 100th epoch, we used the GMM to generate training experiences, while afterward, we switched to using the value function for generating training experiences. 

\subsection{Impact of Varying Timeslots size}

\begin{figure}[t]
	\subfigure[][{\scriptsize Service Rate}]{
		\scalebox{0.25}[0.25]{\includegraphics{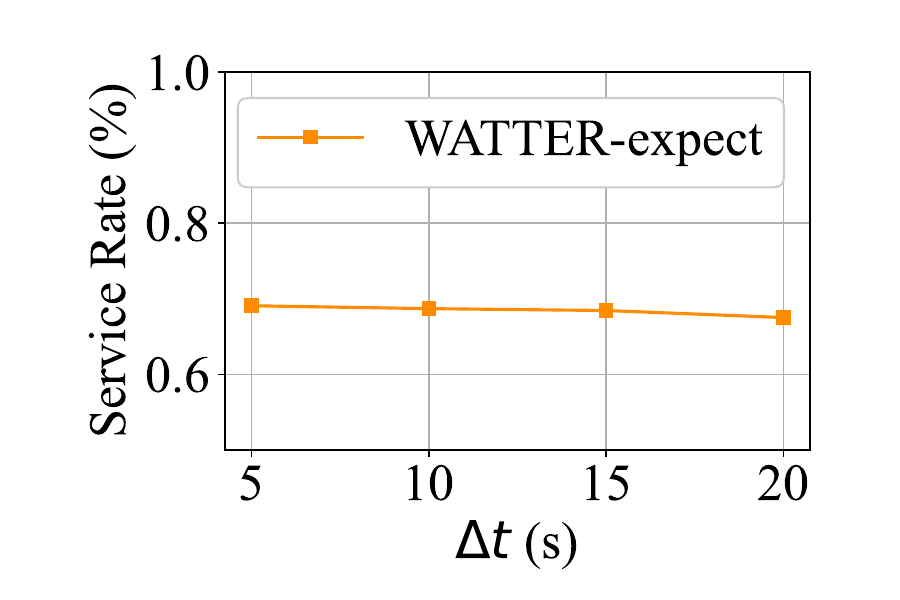}}
		\label{subfig:timeslots_served_rate}}
	\hfill
	\subfigure[][{\scriptsize Average Extra Time}]{
		\scalebox{0.25}[0.25]{\includegraphics{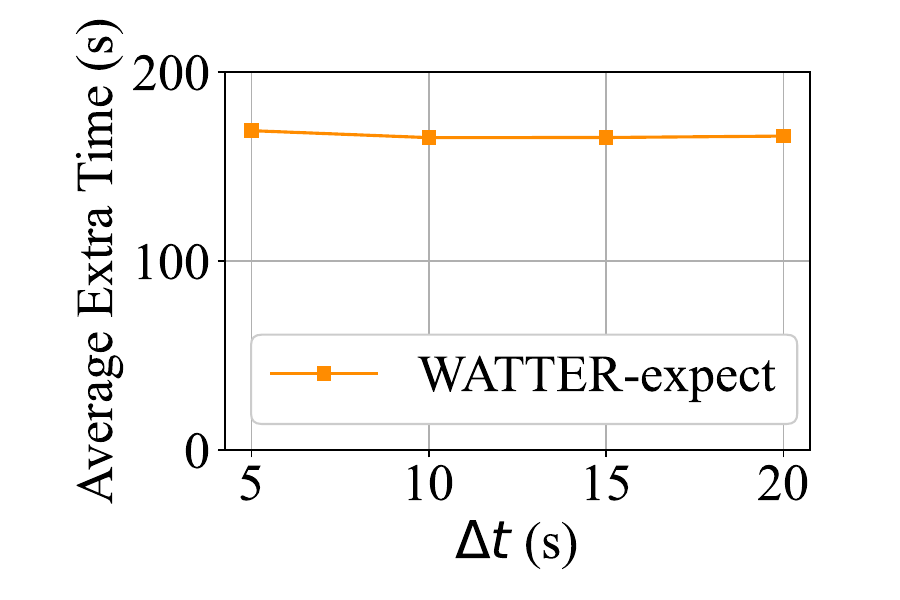}}
		\label{subfig:timeslots_avg_extra}}
	
	\subfigure[][{\scriptsize Unified Cost}]{
		\scalebox{0.25}[0.25]{\includegraphics{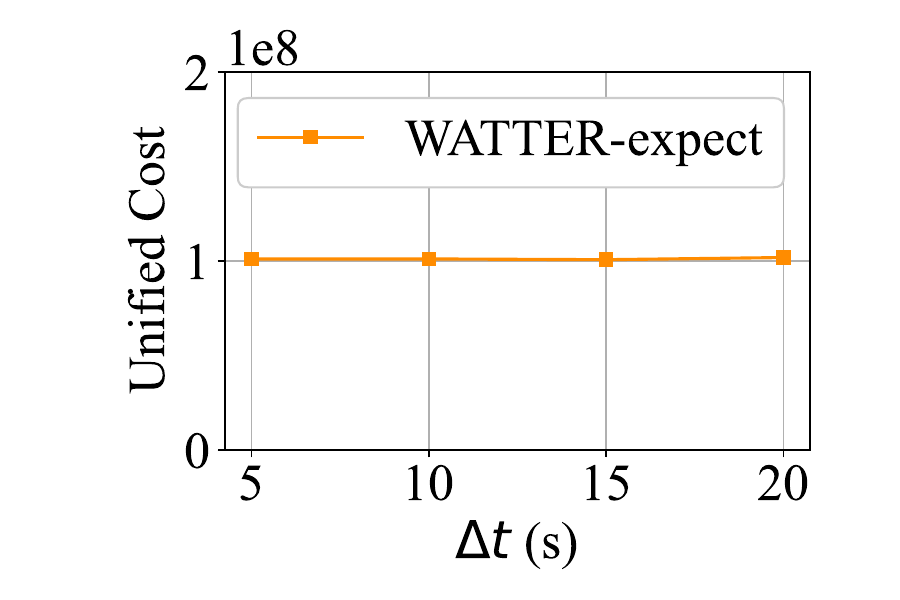}}
		\label{subfig:timeslots_cost}}
	\hfill
	\subfigure[][{\scriptsize Running Time}]{
		\scalebox{0.25}[0.25]{\includegraphics{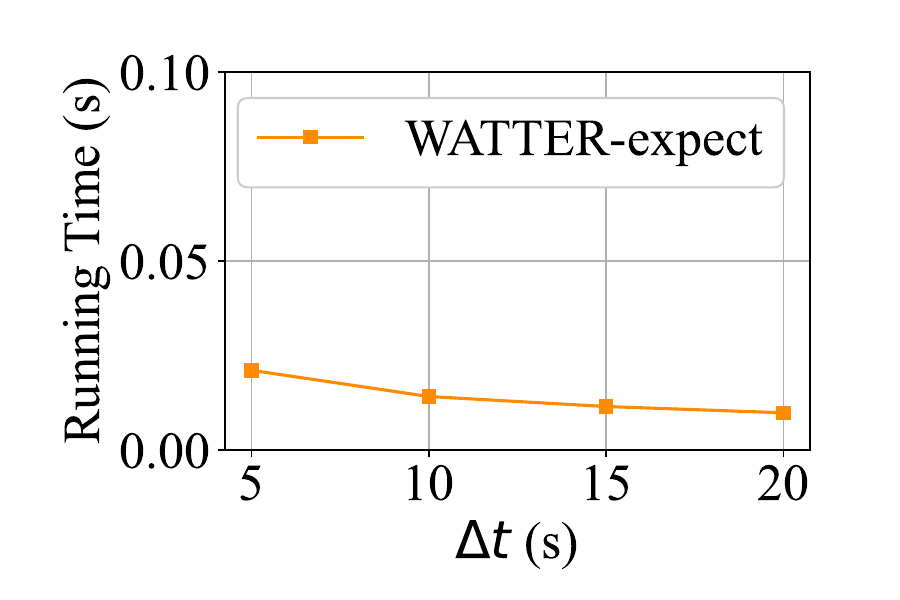}}
		\label{subfig:timeslots_running_time}}
	
	\caption{\small Performance of varying timeslots size on NYC dataset.}
	\label{fig:var_timeslot}
\end{figure}

We conducted experiments on the NYC dataset to investigate the effect of varying time slot size $\Delta t$ and watching window $\eta^{(i)}$. Figure \ref{fig:var_timeslot} illustrates the performance of WATTER-expect when $\Delta t$ is varied from $5s$ to $10s$. We observed no significant effect in the service rate, average extra time, and unified cost after changing $\Delta t$. This can be attributed to two reasons. Firstly, most orders are unlikely to have their best group updated within a short period of time. Thus, the impact of $\Delta t$ on whether orders need to be dispatched is minimal when $\Delta t$ is small. Secondly, we utilized a shallow neural network as the value function. Compared to the geographical information of the orders, the waiting time has a less significant impact on the network. Additionally, we discovered that the running time metric significantly decreases as $\Delta t$ increases. This is due to the decreased frequency of neural network inference and graph traversal.

\subsection{Impact of Varying Watching Window}

In addition, we conducted experiments to observe the impact of the value of the watching window parameter $\eta^{(i)}$ on the performance of the WATTER-expect, WATTER-online, and WATTER-timeout algorithms. As shown in Figure \ref{fig:var_window}, for each order $o^{(i)}$, we set its watching window as $\eta^{(i)} = t^{(i)} + \eta * cost(l_p^{(i)}, l_d^{(i)})$. When varying $\eta$ from $0.5$ to $0.9$, we find that all three algorithms show a significant decrease in service rate, especially WATTER-timeout. This decrease occurred because some orders became impossible to find a shareable route as the waiting window approached.

Considering the average extra time, WATTER-timeout showed a noticeable increase as $\eta$ increased. Since long waiting time makes it harder to further reduce the detour time. In contrast, WATTER-online had relatively small performance changes compared to the other two algorithms under this parameter. Similar trends were observed in the unified cost as well. Lastly, due to the increase in waiting time, the number of orders in the order pool also increased. The overhead of traversal and enumeration resulted in a significant increase in running time for all three algorithms.

\begin{figure}[t]
	\subfigure{
		\scalebox{0.22}[0.22]{\includegraphics{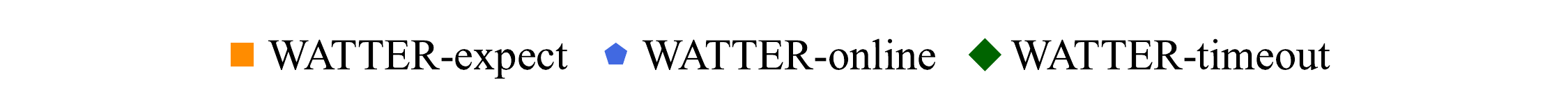}}}\hfill\\\vspace{-3ex}
	\addtocounter{subfigure}{-1}
	
	\subfigure[][{\scriptsize Service Rate}]{
		\scalebox{0.25}[0.25]{\includegraphics{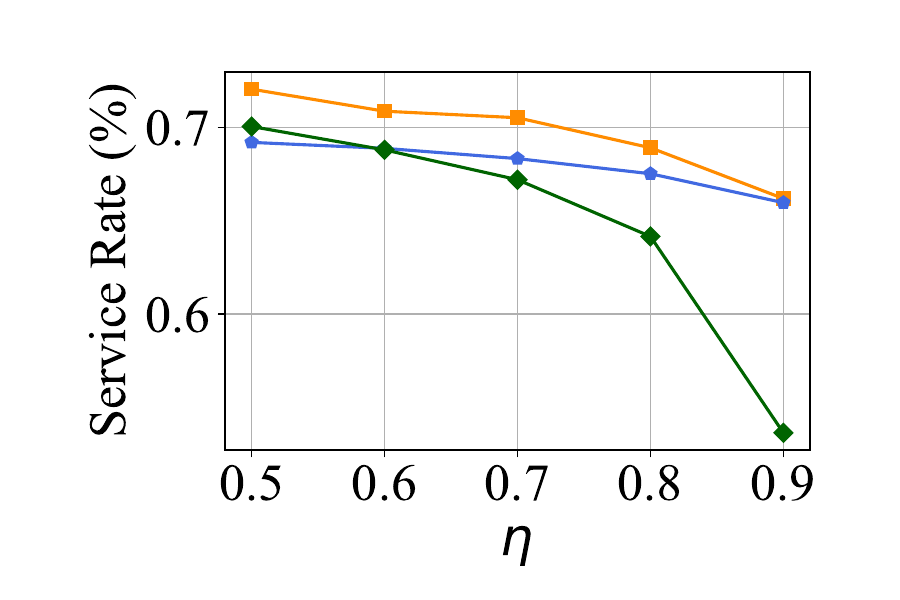}}
		\label{subfig:window_served_rate}}
	\hfill
	\subfigure[][{\scriptsize Average Extra Time}]{
		\scalebox{0.25}[0.25]{\includegraphics{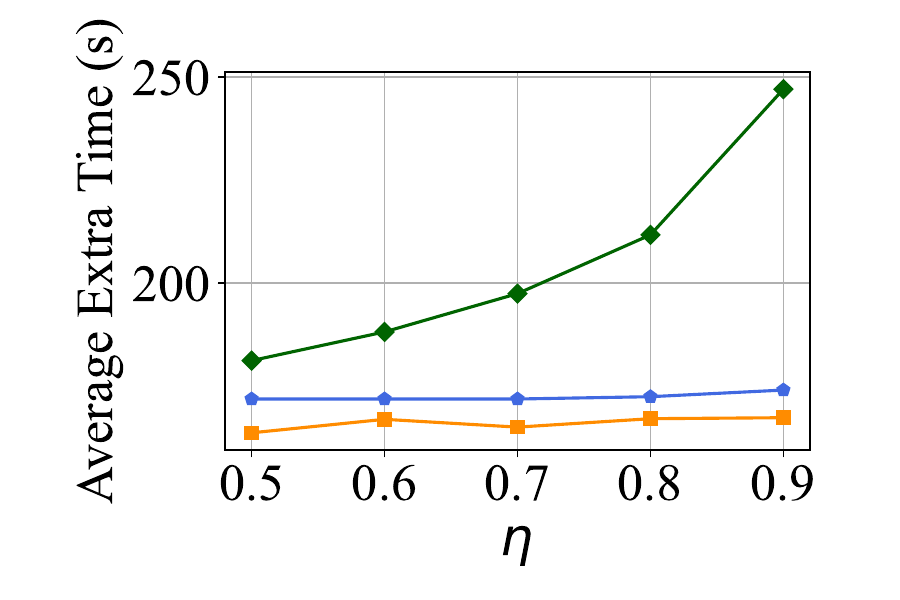}}
		\label{subfig:window_avg_extra}}
	
	\subfigure[][{\scriptsize Unified Cost}]{
		\scalebox{0.25}[0.25]{\includegraphics{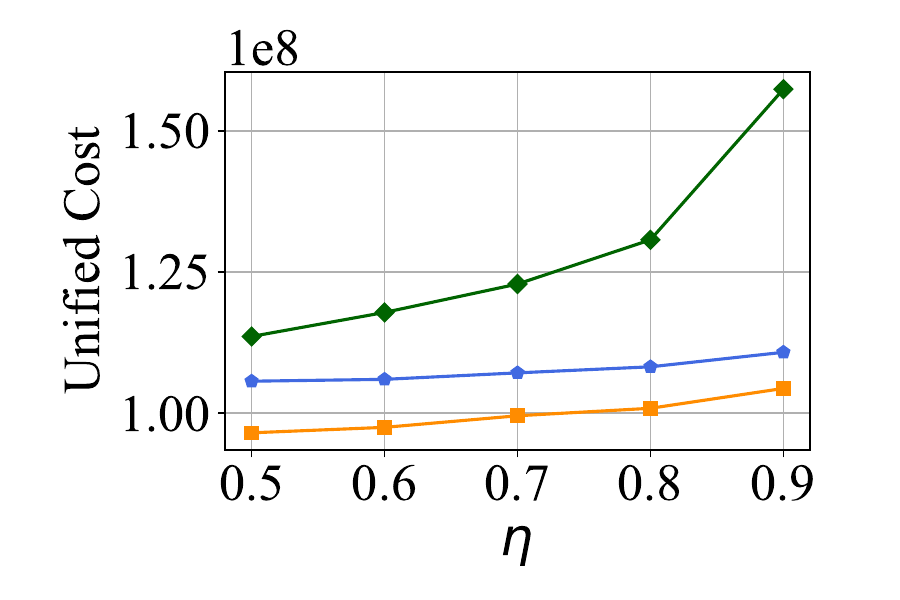}}
		\label{subfig:window_cost}}
	\hfill
	\subfigure[][{\scriptsize Running Time}]{
		\scalebox{0.25}[0.25]{\includegraphics{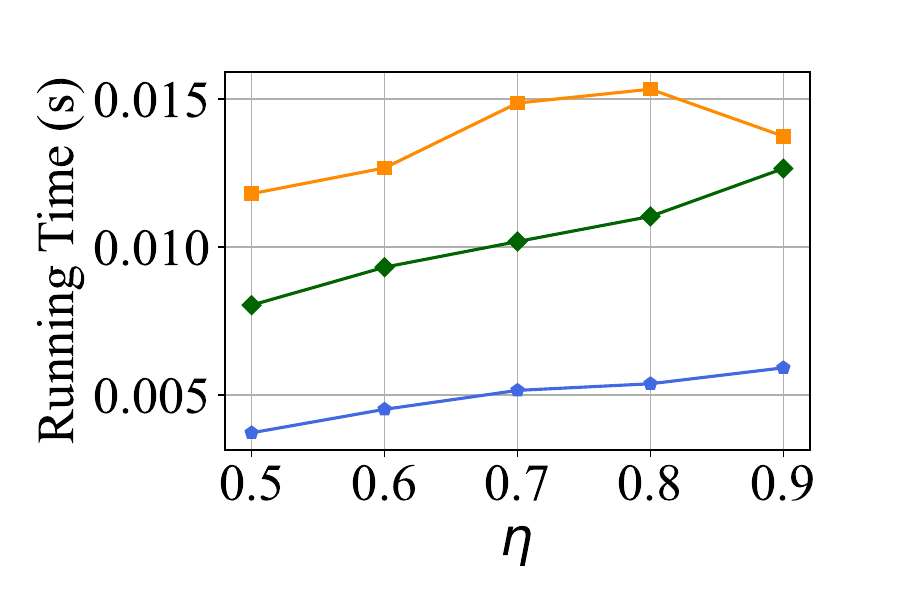}}
		\label{subfig:window_running_time}}
	
	\caption{\small Performance of varying watching window on NYC dataset.}
	\label{fig:var_window}
\end{figure}

\end{document}